\documentclass[a4]{scrartcl}
\usepackage{amsthm}
\usepackage[english]{babel}
\usepackage{amsfonts}
\usepackage{mathtools}
\usepackage{algorithm}
\usepackage[noend]{algpseudocode}
\usepackage[font=normalsize,labelfont=sf,textfont=sf]{subfig}
\usepackage{stfloats}
\usepackage{url}
\usepackage{graphicx}
\usepackage{comment}
\usepackage{xcolor}
\usepackage{booktabs}
\usepackage{siunitx}
\usepackage{stackrel}
\usepackage{enumitem}
\usepackage{pifont}
\usepackage{tablefootnote}
\usepackage{hyperref}
\usepackage{multirow}

\DeclareMathOperator{\R}{{\mathbb{R}}}
\newcommand{\spline}{{\,\varphi}}

\renewcommand{\vec}[1]{\mathbf{#1}}
\newcommand{\mat}[1]{\mathbf{#1}}

\newcommand{\param}{\boldsymbol{\theta}} % param set
\DeclareMathOperator*{\argmin}{arg\,min} %handles subscripts like \lim
\DeclareMathOperator*{\essinf}{ess\,inf} %handles subscripts like \lim
\DeclareMathOperator{\prox}{prox} %handle normal subscripts

\renewcommand{\vec}[1]{\mathbf{#1}}

\newcommand{\N}{\mathbb{N}}

\DeclareMathOperator{\Lip}{Lip}

\newtheorem{theorem}{Theorem}[section]
\newtheorem{lemma}[theorem]{Lemma}
\newtheorem{remark}[theorem]{Remark}
\newtheorem{definition}[theorem]{Definition}

\newtheorem{proposition}[theorem]{Proposition}

\title{Learning Weakly Convex Regularizers for \\Convergent Image-Reconstruction Algorithms}
\date{\today}
\author{Alexis Goujon, Sebastian Neumayer, Michael Unser}

\begin{document}
\maketitle
\renewcommand*{\thefootnote}{\fnsymbol{footnote}}
\footnotetext[1]{The authors are with the Biomedical Imaging Group, \'Ecole polytechnique f\'ed\'erale de Lausanne (EPFL),
	Station 17, CH-1015 Lausanne, {\text \{forename.name\}@epfl.ch}. }
\renewcommand*{\thefootnote}{\arabic{footnote}}

\begin{abstract}
    We propose to learn non-convex regularizers with a prescribed upper bound on their weak-convexity modulus.
    Such regularizers give rise to variational denoisers that minimize a convex energy.
    They rely on few parameters (less than 15,000) and offer a signal-processing interpretation as they mimic handcrafted sparsity-promoting regularizers.
    Through numerical experiments, we show that such denoisers outperform convex-regularization methods as well as the popular BM3D denoiser.
    Additionally, the learned regularizer can be deployed to solve inverse problems with iterative schemes that provably converge.
    For both CT and MRI reconstruction, the regularizer generalizes well and offers an excellent tradeoff between performance, number of parameters, guarantees, and interpretability when compared to other data-driven approaches.
    %is on par with frameworks that rely on deep CNNs with few orders of magnitude more parameters.

    %Since we are able to efficiently control the modulus of weak convexity, the specific case of denoising can be guaranteed to be a convex optimization problem.
    %Hence, we can utilize standard convex optimization techniques, which provably converge to global minima.
    %In our numerical denoising experiments, we outperform the popular BM3D denoiser and several variational approaches that come with convergence guarantees across various noise scales.
    %Inspired by these findings, we also employ the learned regularizer for inverse problems, namely MRI and CT. 
    %As the interplay between forward operators and weak convexity is complicatd, we instead propose to use our regularizer in a PnP fashion, where we can guarantee convergence to stationary points.
    %Our numerical results confirm that our learned regularizer generalized well to generic inverse problems.
\end{abstract}
% 68U10: Image processing
% 47A52: Ill-posed problems, regularization
% 49N45: Inverse problems (Calculus of variations and optimal control; optimization)
% 65D07: splines
% 68T05: Learning and adaptive systems
% 26B25: Convexity, generalizations (Real functions, Functions of several variables)
% 47N10: Applications in optimization, convex analysis, mathematical programming, economics (Miscellaneous applications of operator theory)
% 90C26: Nonconvex programming, global optimization (Mathematical programming)

\section{Introduction}\label{sec:Intro}
Linear inverse problems are ubiquitous in imaging, with applications in medical imaging~\cite{MM2019}, including magnetic-resonance imaging (MRI) and X-ray
computed tomography (CT).
In a discretized linear inverse problem~\cite{ribes2008linear}, the goal is to reconstruct an (unknown) image of interest $\vec x\in\R^d$ from a given noisy observation
\begin{equation}\label{eq:InvProb}
    \vec y = \vec H\vec x + \vec n \in\R^m,
\end{equation}
where $\vec H \in \R^{m\times d}$ denotes the measurement operator and $\vec n\in\R^m$ is a noise term.
To overcome a possibly ill-conditioned $\vec H$ and the presence of noise, it is standard to compute the reconstruction $\hat{\vec x}$ as a solution of the variational problem
\begin{equation}\label{eq:VarProb}
    \hat{\vec x} = \argmin \limits_{\vec x \in \R^d} \frac{1}{2}\Vert \vec H\vec x - \vec y \Vert_2^2 + R(\vec x),
\end{equation}
where the regularizer $R\colon \R^d\rightarrow \R$ incorporates prior information about $\vec x$.
Convex regularizers, such as the Tikhonov \cite{tikhonov1963} or total-variation (TV) \cite{rudin1992nonlinear, donoho2006compressed} ones, are popular as they allow one to efficiently solve \eqref{eq:VarProb}.
Unfortunately, such regularizers do not yield state-of-the-art reconstructions and have known limitations.
For instance, they typically struggle to preserve textures in the image $\vec x$ \cite{Nikolova2015}.
\subsection{The Convex Non-Convex Framework for Denoising}
The reliance on a well-chosen non-convex $R$ leads to improved performance \cite{RotBla2009, CheRan2014}, with the caveat that finding a global minimum of \eqref{eq:VarProb} becomes intractable in general.
A possible remedy is provided by the convex non-convex (CNC) framework.
It consists in the deployment of a non-convex $R_{\mathrm{CNC}}$ such that the global objective
\begin{equation}
    \mathcal{J}(\vec x) = \frac{1}{2}\Vert \vec H\vec x - \vec y \Vert_2^2 + R_{\mathrm{CNC}}(\vec x)
\end{equation}
is convex, see \cite{Lanza2021} for an overview.
Over the past years, the use of CNC approaches has led to improved results in various settings, including dictionary learning \cite{TanLiZha2020}, plug-and-play (PnP) algorithms \cite{LiLiXie2021,HurLec2022}, and matrix completion \cite{Abe2020}.

For the case of image denoising, namely $\mat H = \mat I$, the data-fidelity term $\frac12\|\vec x-\vec y\|^2_2$ is \mbox{$1$-strongly} convex.
Hence, to ensure the convexity of $\mathcal{J}$, the regularizer $R_{\mathrm{CNC}}$ needs to be $1$-weakly convex from the definition of weak convexity (see Section \ref{sec:WeakConv}).
Various strategies have been proposed to design a weakly convex $R_{\mathrm{CNC}}$.
\paragraph{Explicit Design}
The commonly used $\Vert \cdot \Vert_1$-norm for sparse regularization can be replaced by a non-convex penalty function that better mimics the behavior of the $\Vert \cdot \Vert_0$-norm, while ensuring that one remains within the CNC framework.
This includes properly scaled versions of the logarithm and the minimax concave penalty \cite{LanMorSel2019}.
Although non-convex, these functions are quasi-convex. In particular, they are such that large values are more penalized than smaller ones.
The potentials are then combined with convolutional filters.
This yields, for instance, TV-like regularizers \cite{ZSZLLD2019, ScrChoPes2022}, which extend and improve upon their convex counterparts.
 \paragraph{Implicit Design with Moreau Envelopes}
 There is a systematic method to convert any convex regularizer into a non-convex but still 1-weakly convex one, utilizing its (generalized) Moreau envelope \cite{Abe2020,Lanza2021}.
Such a regularizer, however, does not admit a closed form, and the existing algorithms to solve \eqref{eq:VarProb} involve a computationally intensive bilevel optimization task.
 \paragraph{Implicit Design via the Learning of Proximal Operators}
Although $R_{\mathrm{CNC}}$ is non-convex, its proximity operator $\prox_{R_{\mathrm{CNC}}}$ is well-defined under mild conditions \cite{GriNik2020}.
In \cite{HurLec2022}, the authors propose to directly learn $\prox_{R_{\mathrm{CNC}}}$ such that it is a good Gaussian denoiser.
To do so, they explicitly parameterize the proximal operator, in line with the recently introduced gradient-step denoisers \cite{cohen2021has,hurault2022gradient}.
More precisely, they express the residual map $(\prox_{R_{\mathrm{CNC}}} - \mathrm{Id})$ as the gradient of a deep convolutional neural network (CNN), and require that the residual is contractive by enforcing that it has a Lipschitz constant smaller than 1.
This yields excellent performance, with the caveat that it is challenging to enforce strict Lipschitz constraints on the gradient of a CNN.
For this reason, the authors of \cite{HurLec2022} propose to regularize the spectral norm of the Jacobian of $(\prox_{R_{\mathrm{CNC}}} - \mathrm{Id})$ at a finite number of locations.
This method works well in practice but does not offer any provable guarantee on the weak-convexity property of the underlying (implicit) objective $\mathcal{J}$.

\subsection{Extension to Ill-Posed Inverse Problems}
The design of CNC models is difficult when the forward matrix $\mat H$ is noninvertible.
Since the data term $\frac{1}{2}\Vert \vec H\vec x - \vec y \Vert_2^2$ is not strongly convex anymore, the $1$-weak convexity of $R$ is not sufficient.
%The data term $\frac{1}{2}\Vert \vec H\vec x - \vec y \Vert_2^2$ might behave very differently for different directions, namely it can be $\spline$-strongly convex with large $\spline$ in one direction, while only being constant in others.
Then, the condition on $R$ depends on $\mat H$, and CNC models are therefore usually tailored to a specific problem.
One can partially circumvent this limitation by combining a proximal algorithm with a generic weakly convex regularizer, for which the proximal operator is well-defined.
%If $R$ is only known implicitly in terms of its proximal operator, this approach is commonly referred to as PnP methods \cite{venkatakrishnan2013plug}.
The convergence to stationary points of the objective is established in \cite{HurLec2022,HurChaLec2023} for the forward-backward splitting \cite{beck2009fast} based on the very general convergence result for functions with the Kurdyka-Łojasiewicz (KL) property given in \cite{AttBolSva2013}.
When $R$ is differentiable, as will be assumed in our setting, similar results can be obtained for gradient descent applied to the non-convex objective \eqref{eq:VarProb}, see \cite{AttBolSva2013}.
From a stochastic perspective, it is known that such first-order methods do not get trapped into strict saddle points of the objective \cite{Lee2019}.
This is a possible explanation for the good empirical performance of non-convex reconstruction frameworks.

\subsection{Other Deep-Learning-Based Variational Methods with Some Guarantees}
The emergence of deep-learning-based methods has led to significant improvements in the quality of reconstruction for inverse problems.
Yet, due to the blackbox nature of deep NNs, this often comes with a loss of interpretability and reliability.
Thus, there is a growing interest to mitigate these limitations, see \cite{MukHauOek2023} for a survey.
In the following, we briefly comment on works that rely on the variational formulation \eqref{eq:VarProb} with a learned regularizer $R$ but that are not directly within the CNC framework.
To provide maximal theoretical guarantees within iterative image reconstruction, it was proposed in \cite{MukDit2021} to learn a convex $R$ based on a deep CNN, and shown in \cite{GouNeuBoh2022} that a shallow model, namely, a convex ridge regularizer NN (CRR-NN) with few parameters, was sufficient.
The latter offers the opportunity to learn a collection of filters and sparsity-promoting profile functions to build $R$.
This is inspired by the Fields-of-Experts (FoE) framework \cite{RotBla2009} and its many variants, such as \cite{CheRan2014}, to design and learn a non-convex $R$.
While \cite{CheRan2014} yields good performance, it does not guarantee that the objective is convex.
Another popular extension of FoE is trainable nonlinear reaction diffusion (TRND)~\cite{chen2016trainable}.
There, the minimization scheme associated with \eqref{eq:VarProb} is unrolled and different filters and potential functions are learned at each step.
This improves the performance over \cite{CheRan2014} but does not correspond to an energy minimization anymore.
More recently, all these frameworks have been unified in the context of variational networks \cite{KKHP2017}.
The combination of these with recent findings in deep CNN research and early stopping techniques has then led to the total deep variation framework \cite{KobEff2020}.
Although this model has several layers, some interpretability remains possible through an eigenfunction analysis.
Another deep-learning-based variational method with convergence guarantees and a regularization scheme is found in~\cite{LiSch2020}.

\subsection{Outline and Main Contributions}
In this work, we propose a framework to learn a $1$-weakly convex regularizer that yields an interpretable proximal denoiser.
The general framework is introduced in Section~\ref{sec:WeakConv}.
Then, the principal contributions are as follows.
\begin{itemize}
    \item{\textbf{Denoising}:} In Section~\ref{sec:Denoising}, we propose a scheme for the training of weakly-convex-ridge-regularizer neural networks (WCRR-NN), with a significant increase in performance over their convex counterparts but with the same guarantees and interpretability. 
    Based on a condition introduced in Proposition~\ref{pr:weakcvxconstantBound}, the associated denoising problem is convex, which allows for global minimization.
    Numerical experiments indicate that the learning of both the profiles and the filters leads to a sparsity prior that is state-of-the-art in the CNC framework across various noise levels for the BSD68 test set.
    In particular, it is the first convex-energy-based model that outperforms BM3D \cite{DabFoiKat2007}, which has been one of the most popular benchmarks for nearly 15 years now.
    \item{\textbf{Inverse problems}:} In Section~\ref{sec:InverseProblems}, we deploy the learned regularizer to solve generic inverse problems by minimizing \eqref{eq:VarProb} with an accelerated gradient-descent (AGD) scheme that is tailored to our weakly convex regularizer (Algorithm \ref{alg:AGD}).
    Further, we prove that the algorithm reaches some critical point of the objective (Theorem~\ref{thm:ConvGD}).
    Numerical experiments for CT and MRI demonstrate that the regularizer empirically generalizes well. We find that it outperforms several energy-based reconstruction methods that come with convergence guarantees.
    
\end{itemize}
Finally, conclusions are drawn in Section~\ref{sec:Conclusions}.
The implementation of WCRR-NNs and pre-trained models are publicly available\footnote{\url{https://github.com/axgoujon/weakly_convex_ridge_regularizer}} as well as their usage to solve inverse problems\footnote{\url{https://github.com/axgoujon/convex_ridge_regularizers}}.

\section{Weakly Convex Regularizers}\label{sec:WeakConv}
Our goal is to construct a regularizer $R$ for the variational reconstruction model \eqref{eq:VarProb} that performs well across a variety of inverse problems, while maintaining the theoretical guarantees and interpretability of classical schemes.
A particularly promising direction is given by the CNC framework, where one can efficiently find a global minimum of the objective in \eqref{eq:VarProb}.
As commonly done in practice, our strategy is to design and train the regularizer based on the denoising task
\begin{equation}
    \label{eq:denoise}
    \hat{\vec{x}} = \argmin\limits_{\vec{x}\in\R^d} \frac{1}{2}\|\vec{x} - \vec{y}\|_2^2 + R(\vec{x}),
\end{equation}
where $\vec y$ is a noisy version of a clean image.
The minimization of \eqref{eq:VarProb} for generic inverse problems and weakly convex regularizers is then discussed in Section~\ref{sec:InverseProblems}.

To obtain a CNC model in \eqref{eq:denoise}, $R$ needs to be $1$-weakly convex so that the overall objective remains convex.
\begin{definition}
A function $f\colon \R^d \to \R$ is
\setlist[enumerate]{leftmargin=14mm, label=\roman*)}
\begin{enumerate}
    \item \emph{convex} if $f(\lambda \vec x + (1-\lambda)\vec y) \leq \lambda f(\vec x) + (1-\lambda)f(\vec y)$ for all $\vec x,\vec y \in \R^d$ and $\lambda \in [0,1]$;
    \item \emph{$\rho$-strongly convex} if $(f - \frac{\rho}{2} \Vert \cdot \Vert^2)$ is convex with $\rho\geq0$;
    \item \emph{$\rho$-weakly convex} if $f + \frac{\rho}{2} \Vert \cdot \Vert^2$ is convex with $\rho\geq0$.
\end{enumerate}
\end{definition}
Note that a $\rho$-weakly convex $R$ is also $\mu$-weakly convex for any $\mu \geq \rho$.
A convex $R$ is $\rho$-weakly convex for any $\rho\geq 0$ and, in particular, $0$-weakly convex.
For a differentiable $R$, convexity is equivalent to the monotonicity of $\boldsymbol{\nabla }R$.
Hence, a differentiable $R$ is $\rho$-weakly convex iff
\begin{equation}
\label{eq:wcvx_monotone}
    (\boldsymbol{\nabla}R(\vec y) - \boldsymbol{\nabla}R(\vec x))^T(\vec y - \vec x) \geq -\rho \|\vec y - \vec x\|_2^2,
\end{equation}
for any $\vec x, \vec y \in\R^d$.
Given a twice-differentiable $R$, $\rho$-weak convexity is equivalent to
\begin{equation}
    H_R(\vec x) \succeq -\rho \mathbf{I}
\end{equation}
for any $\vec x\in\R^d$, where $H_R(\vec x)$ denotes the Hessian of $R$ at $\vec x$.
In other words, the Hessian of a $\rho$-weakly convex function has all its eigenvalues in the range $[-\rho, + \infty)$.
\begin{remark}
\label{rk:Lip_implies_wcvx}
Any differentiable function $R$ with \mbox{$L$-Lipschitz} gradient is \mbox{$L$-weakly} convex.
This estimate is, however, not necessarily tight, in the sense that $R$ might also be $\rho$-weakly convex for some $0\leq \rho \ll L$ all the way to zero.
For instance, any convex $R$ with $L$-Lipschitz gradient is $L$-weakly convex but it is also trivially $0$-weakly convex because it is equivalent to being convex.
%A more precise estimate will be derived for our choice of $R$ to make the most of the additional flexibility offered by weak convexity (Proposition~\ref{pr:weakcvxconstantBound}).
\end{remark}
Weak convexity provides more flexibility, while still maintaining most of the desirable properties of usual convex-regularization frameworks.
In particular, the proximal operator
\begin{equation}\label{eq:prox}
\prox_{R}(\vec y) = \argmin\limits_{\vec{x}\in\R^d} \frac{1}{2}\|\vec{x} - \vec{y}\|_2^2 + R(\vec{x})
\end{equation}
is well-defined for any $\rho$-weakly convex $R$ with $\rho < 1$.
Indeed, the objective in \eqref{eq:prox} is \mbox{$(1 - \rho)$-strongly} convex, which ensures the existence of a unique minimizer.
The properties of the proximal operator in a generic non-convex setting are characterized in detail in \cite{GriNik2020}. The main implication here is the Lipschitz continuity of our denoiser \eqref{eq:prox} (Proposition~\ref{prop:ProxCNC}).
%summarize in Theorem~\ref{thm:ProxCNC} the results relevant to the present work.
\begin{proposition}[\cite{GriNik2020}]\label{prop:ProxCNC}
For any $\rho$-weakly convex regularizer $R$ with $\rho <1$, there exists a convex lower semi-continuous potential $g \colon \R^d \to \R$ such that $\prox_{R}(\vec x) \in \partial g(\vec x)$ holds for every $\vec x \in \R^d$.
Conversely, the subgradient of any such $g$ coincides with $\prox_R$ for some $R$ that is $1$-weakly convex on any convex subset of its domain.
Furthermore, $\prox_{R}$ is $(\frac{1}{1-\rho})$-Lipschitz, in the sense that
\begin{equation}
    \|\prox_{R}(\vec y_2) - \prox_{R}(\vec y_1)\|_2 \leq \frac{1}{1-\rho} \|\vec y_2 - \vec y_1\|_2
\end{equation}
for any $\vec y_1, \vec y_2 \in\R^d$.
More generally, $C^{k+1}$ regularity of the potential $g$ leads to $C^k$ regularity of $\prox_{R}$.
Finally, $\prox_{R}$ is invertible on its range in this setting.
\end{proposition}
For a convex regularizer $R$, $\rho=0$ and, hence, $\prox_R$ is non-expansive ($1$-Lipschitz).
For a non-convex but weakly convex $R$, namely $\rho>0$, this is not necessarily the case anymore.
We conjecture that this is key to the boost in performance since non-expansive denoisers have intrinsic limitations, see for instance \cite[Fig.~1]{GouNeuBoh2022}.

\section{Design of a Learnable and Provably 1-Weakly Convex Regularizer for Denoising}\label{sec:Denoising}
In this section, we discuss the construction and training of $R$, and compare it with several other variational frameworks.
\subsection{Regularizer Architecture}\label{sec:Arch}
The weakly convex regularizer $R$ is chosen as the sum of convolutional ridges
\begin{equation}
\label{eq:convconvexridgereg}
    R\colon x[\cdot] \mapsto \sum_{i=1}^{N_C}\sum_{\vec k\in\R^2}\psi_i\bigl((h_i * x)[\vec k]\bigr),
\end{equation}
where $x[\cdot]$ represents a 2D image, $(h_i[\cdot])_{i=1}^{N_C}$ are the impulse responses of a collection of linear and shift-invariant filters, and $(\psi_i)_{i=1}^{N_C}$ are potential functions with Lipschitz continuous derivative.
In practice, the finite-size input images are zero-padded so that the outputs of the convolutions have the same spatial size as the input image.
We also choose potential functions with a shared profile $\psi$, so that $\psi_i= \alpha_i^{-2}\psi(\alpha_i\cdot)$ with $\alpha_i>0$. As the $h_i[\cdot]$ can absorb the $\alpha_i$ in the definition of $\psi_i$, this is just a different parameterization for adding weights $\alpha_i^{-2}$ in front of the profile $\psi$.
The advantage of our parameterization is that $\Lip(\psi_i')$ does not depend on $\alpha_i$, which will simplify the reasoning throughout this section. The number $N_C$ of filters is also referred to as the number of channels or feature maps of the model.
The motivations behind our choice are threefold.
\begin{itemize}
    \item {\textbf{Interpretability}:} The model \eqref{eq:convconvexridgereg} includes many traditional sparsity-promoting regularizers and, as shown, in Section~\ref{sec:UnderHood}, the trained regularizer will have a simple signal-processing interpretation.
    The parameters for deeper CNN-based regularizers are usually much harder to interpret than those in \eqref{eq:convconvexridgereg}.
    \item {\textbf{Control of $\rho$}:} The weak-convexity modulus of \eqref{eq:convconvexridgereg} can be upper-bounded using Proposition~\ref{pr:weakcvxconstantBound}.
    This is far less obvious for deeper CNNs architectures.
    There, weak convexity is usually promoted via regularization during training \cite{HurLec2022}.
    While this works qualitatively, it does not generate provably $\rho$-weakly convex maps for some prescribed~$\rho$.
    \item {\textbf{Model expressivity}:} There is evidence that, in constrained settings, \eqref{eq:convconvexridgereg} has good expressive power.
For instance, when learning convex regularizers for \eqref{eq:VarProb}, architectures of the form \eqref{eq:convconvexridgereg} are on par with deep CNNs such as the input convex NN (ICNN), all the while depending on much fewer parameters \cite{GouNeuBoh2022}.
\end{itemize}
To simplify the notation in the sequel, the regularizer \eqref{eq:convconvexridgereg} is written whenever needed in the generic form
\begin{equation}
\label{eq:convexridgereg}
    R\colon\vec{x}\mapsto \sum_{j=1}^{d \times N_C} \psi_j(\vec{w}_j^T \vec{x}),
\end{equation}
where $\vec x=(x[\vec k])_{\vec k\in \Omega} \in\R^d$ is the vectorized representation of $x[\cdot]$, the $\vec{w}_j\in\R^d$ correspond to shifted versions of the filter kernels, and $j$ indexes at the same time along the channels and the 2D shifts of the kernels.
The gradient of this differentiable regularizer $R$ reads
\begin{equation}
\label{eq:gradmodel}
    \boldsymbol{\nabla} R(\vec x) = \mat W^T \boldsymbol{\spline}(\mat W \vec x),
\end{equation}
where $\mat W=[\vec w_1 \cdots \vec w_{dN_C}]^T \in\R^{dN_C\times d}$ and $\boldsymbol{\spline}$ is the pointwise \textit{activation function} given by $\boldsymbol{\spline}(\vec z) = (\psi_j'(z_j))_{j=1}^{dN_C} = (\alpha_j^{-1}\psi'(\alpha_j z_j))_{j=1}^{dN_C}$.
Note that $\mat W \vec x$ is a multichannel filtered version of the image $\vec x$.
Since $\psi$ can absorb the spectral norm of $\mat W$, we enforce that $\|\mat W\| = 1$, where $\|\cdot\|$ denotes the spectral norm, in order to remove some redundancy from the model and simplify the explanations.

In the following, we use that the Lipschitz continuity of $\psi^\prime$ implies differentiability of $\psi^\prime$ almost everywhere (Rademacher's theorem) and that the essential infimum $\essinf_{t\in\R} \psi''(t)$ is well-defined and satisfies $\vert \essinf_{t\in\R} \psi''(t) \vert \leq \Lip(\psi')$.
\begin{lemma}
\label{lm:wcvx_1d}
    Let $\psi\colon \R \to \R$ have a Lipschitz continuous derivative.
    Then $\psi$ is $\rho$-weakly convex for any $\rho\geq s_{\inf} = \max(0, -\essinf_{t\in\R} \psi''(t))$.
\end{lemma}
\begin{proof}
    The Lipschitz continuity of $\psi'$ implies that $\psi'(t_2) - \psi'(t_1) = \int_{t_1}^{t_2}\psi''(t)\mathrm{d}t$ for any $t_1, t_2\in\R$.
    From this, we infer that $(\psi(t_2) - \psi(t_1))(t_2 - t_1) \geq (\essinf_{t\in\R} \psi''(t))(t_2 - t_1)^2$, which is precisely condition \eqref{eq:wcvx_monotone}.
\end{proof}
\begin{proposition}
\label{pr:weakcvxconstantBound}
Any $R$ of the form \eqref{eq:convexridgereg} with $\|\mat W\| = 1$ and a $\rho$-weakly convex $\psi$ is $\rho$-weakly convex.
In particular, assuming that $\psi'$ is Lipschitz continuous, this holds for any $\rho \geq s_{\inf}$ as defined in Lemma \ref{lm:wcvx_1d}.
\end{proposition}
\begin{proof}
Since $\alpha_i>0$ and $\psi_i = \alpha_i^{-2} \psi(\alpha_i \cdot)$, the convexity of $t\mapsto \psi(t) + \frac{\rho}{2}t^2$ implies the convexity  of $t \mapsto \psi_i(t) + \frac{\rho}{2}\alpha_i^{-2}(\alpha_i t)^2$.
Thus, $\vec x \mapsto \psi_j(\vec w_j^T \vec x) + \frac{\rho}{2}(\vec w_j^T \vec x)^2$ and $\vec x \mapsto R(\vec x) + \frac{\rho}{2}\|\mat W \vec x\|_2^2$ are also convex.
Since $\|\mat W \| = 1$ and $\rho>0$, $\vec x \mapsto \frac{\rho}{2}(\|\vec x\|_2^2 - \|\mat W \vec x\|_2^2)$ is convex, and we infer that $\vec x \mapsto R(\vec x) + \frac{\rho}{2} \|\vec x\|_2^2$ is convex.
\end{proof}
Hence, we can obtain a $1$-weakly regularizer $R$ by enforcing that $s_{\inf}\leq 1$.
\begin{remark}
\label{rk:connectioncrrnn}
    The ridge decomposition \eqref{eq:convexridgereg} of $R$ is also used within the CRR-NN framework \cite{GouNeuBoh2022}, which involves the learning of a convex-ridge regularizer $R$ with learnable spline potentials $\psi_j$. For CRR-NNs, $s_{\inf} = 0$ is enforced to ensure that the $\psi_j$ are convex. On the contrary, the present WCRR-NN model with $s_{\inf}\in [0, 1]$ has more freedom and therefore extends upon \cite{GouNeuBoh2022}.
\end{remark}
The present parameterization of $R$, is greatly inspired by \cite{GouNeuBoh2022}.
However, instead of its single (non-decreasing) spline non-linearity used in \cite{GouNeuBoh2022}, we decompose the activation $\spline = \psi^\prime$ into the difference of two splines as
\begin{equation}\label{eq:decomp_profile}
    \spline = \mu\spline_{+} - \spline_{-},
\end{equation}
where $\mu \in\R_{\geq 0}$ is a learnable parameter and the $\spline_{+}, \spline_{-}$ are trainable, non-decreasing, non-expansive linear splines.
Although theoretically equivalent to the use of a single linear spline $\spline$ with $s_{\inf}\leq 1$, we found the decomposition \eqref{eq:decomp_profile} to be more effective for the training.
Theoretical motivations for using splines in a constrained NN have been proposed in \cite{NGBU2022}, and a discussion of the expressivity of the resulting NN architecture can be found in \cite{goujon2022role}.
Our choice ensures that the following properties are met:
\begin{itemize}
    \item $R$ is $1$-weakly convex, which follows from Proposition~\ref{pr:weakcvxconstantBound};
    \item the Lipschitz constant is bounded as
    \begin{equation}\label{eq:estLip}
        \Lip(\boldsymbol{\nabla} R) \leq \|\mat W\|^2 \Lip(\varphi) \leq \max(\mu, 1).
    \end{equation}
    
\end{itemize}
In the following, we provide more parameterization details regarding the parameterization.
\paragraph{Parameterization of Learnable Linear Splines}
Both linear splines $\spline_{+}$ and $\spline_{-}$ are parameterized in the same way with our spline toolbox \cite{BCGA2020,DGB2022}.
In the sequel, we abbreviate their respective learnable parameters $\vec c_+$ and $\vec c_-$ by $\vec c$.
We use $\spline_{\vec c}\colon\mathbb{R}\rightarrow\mathbb{R}$ with knots $\tau_m = (m - M/2)\Delta$, $m=0,\ldots,M$, where $\Delta$ is the spacing.
For simplicity, we assume that $M$ is even.
The learnable parameter $\vec c=(c_m)_{m=0}^{M}\in\mathbb{R}^{M+1}$ defines the values $\spline_{\vec c}(\tau_m)=c_m$ of $\spline_{\vec c}$ at the knots.
To fully characterize $\spline_{\vec c}$, we extend it by the constant value $c_0$ on $(-\infty,\tau_0]$ and $c_M$ on $[\tau_M, +\infty)$.
Consequently, any primitive $\psi$ of $\spline_{\vec c}$ is piecewise quadratic on $[\tau_0, \tau_M]$ with affine extensions.
\paragraph{Constraints on the Linear Splines}
To ensure that the $\spline_i$ are non-decreasing and non-expansive, we follow the strategy introduced in \cite{GouNeuBoh2022}.
Let $\mat D\in \mathbb{R}^{M\times(M+1)}$ be the one-dimensional finite-difference matrix with $(\mat D \vec c)_m=(c_{m+1} - c_m)$ for $m=1, \ldots,M$.
As $\spline_{\vec c}$ is piecewise-linear, it holds that
\begin{equation}\label{eq:FeasSet}
    \text{$\spline_{\vec c}$ is non-decreasing and non-expansive} \Leftrightarrow 0\leq (\mat D \vec c)_m\leq \Delta, \; m=1,\ldots,M.
\end{equation}
To optimize over $\{\spline_{\vec c}\colon 0\leq (\mat D \vec c)_m\leq \Delta, \; m=1,\ldots,M\}$, we reparameterize the linear splines as $\spline_{\boldsymbol{P}_{\uparrow}(\vec c)}$, where
\begin{equation}
\label{eq:projection}
    \boldsymbol{P}_{\uparrow}(\vec c) =  \mat S \mathrm{Clip}_{[0, \Delta]}(\mat D\vec c) + \vec 1^T \vec c
\end{equation}
is a nonlinear projection onto the feasible set \eqref{eq:FeasSet}.
In \eqref{eq:projection}, $\mathrm{Clip}_{[0, \Delta]}$ is the pointwise clipping operation with $\mathrm{Clip}_{[0, \Delta]}(t)=\min(\max(0,t),\Delta)$, and $\mat S$ denotes the cumulative-sum operation with $(\mat S \vec d)_{m+1} = \sum_{k=1}^m d_k$ for $m=0, \ldots, M$ and any $\vec d \in\R^{m}$.
In words, $\boldsymbol{P}_{\uparrow}$ clips the finite differences between entries in $\vec c$ that are either greater than $\Delta$ or negative, and sets them to the closest admissible value, while it preserves the mean due to the additional term $\vec 1^T \vec c$.

Further, we enforce that $\spline_+$ and $\spline_-$ are odd, which is natural for imaging as it results in even potentials.
To get this symmetry while still satisfying \eqref{eq:FeasSet}, we use the change of variable $\vec c\rightarrow \frac{1}{2}(\boldsymbol{P}_{\uparrow}(\vec c) - \texttt{reverse}(\boldsymbol{P}_{\uparrow}(\vec c))$, where $\texttt{reverse}$ flips the order of the entries of $\vec c$.
Hence, all constraints are embedded into the parameterization, and the parameter $\vec c$ that is learned remains unconstrained.
\paragraph{Parameterization of Convolutional Filters}
The learnable convolution layer $\mat W$ is required to be of unit norm.
Hence, we parameterize $\mat W$ as $\mat W=\mat U/\|\mat U\|$, where $\mat U$ represents a convolutional layer with the same dimensions as $\mat W$.
The computation of the spectral norm $\|\mat U\|$ will be described in Section \ref{sec:TrainingProc}.
To efficiently explore a large field of view, see also \cite{GouNeuBoh2022}, we decompose $\mathbf{U}$ into a composition of three zero-padded convolutions with kernels of size $(k_s\times k_s)$, $k_s$ odd, and an increasing number of output channels.
Similarly to \cite{CheRan2014}, the convolution kernels are constrained to have zero mean.
The equivalent (up to boundary effects) \textit{single-convolution} layer would have a kernel of size $(K_s \times K_s)$ with $K_s=3k_s -2$.
%Since WCRNNs are shallow by design, this composition can be interpreted as a compensation for depth.

\subsection{Multi-Noise-Level Denoiser}
\label{subsec:multinoisedenoiser}
So far, we only introduced a generic $R$ that is not adapted to diverse noise levels.
To obtain a denoiser for various noise levels $\sigma$, a common approach is to incorporate an adjustable parameter $\lambda_{\sigma}\in\R$ as in
\begin{equation}\label{eq:denoiserlmbd}
    \hat{\vec x} = \argmin \limits_{\vec x \in \R^d} \frac{1}{2}\Vert \vec x - \vec y \Vert_2^2 + \lambda_{\sigma} R(\vec x).
\end{equation}
In principle, this leads to a noise-level-dependent regularizer $R_{\sigma} = \lambda_{\sigma}R$, but this dependence on $\sigma$ turns out to be too simple to ensure good performance across multiple noise levels.
Another limitation in this setting is that $R_{\sigma}$ is $\lambda_{\sigma}$-weakly convex.
Hence, for $\lambda_\sigma>1$, one is not guaranteed to remain within the CNC framework and, for $\lambda_\sigma<1$, we might not exploit the full freedom given by CNC models.
Therefore, we instead express the parameters $\alpha_i$\footnote{In preliminary investigations, we also attempted the learning of the parameter $\mu$ as a function of the noise level but it did not improve performance. Hence, $\mu$ is chosen to be constant across noise levels.} introduced in Section~\ref{sec:Arch} as functions of the noise level
\begin{align}
    \alpha_i(\sigma) &= \mathrm{e}^{s_{\alpha_i}(\sigma)}/(\sigma + \epsilon),\label{eq:para_sig}
\end{align}
where we set $\epsilon=\num{1e-5}$ to prevent instabilities for small $\sigma$.
Here, $s_{\alpha_i}$ is a learnable linear spline with underlying parameter $\vec c_{\alpha}^i$, which is parameterized similarly to $\varphi_{\pm}$ but without constraints.
The exponential parameterization in \eqref{eq:para_sig} allows for efficiently exploring a large range at training and such a scheme is quite common in learning, e.g. in\ the popular TNRD framework \cite{chen2016trainable}. The scaling by $\sigma$ in \eqref{eq:para_sig} allows for normalizing the noise distribution before the activation and was found to be very helpful in practice.
Ultimately, our noise-level-dependent profile functions \;${\psi_i}(t, \sigma) = (1/\alpha_i(\sigma))^2 \psi(\alpha_i(\sigma) t)$ satisfy
\begin{equation}
    \frac{\partial^2 \psi_i}{\partial t^2}(t, \sigma) = \mu \varphi_{+}'(\alpha_i(\sigma) t) - \varphi_{-}'(\alpha_i(\sigma) t) \in [-1, +\infty).
\end{equation}
%\begin{remark}
%    In the single activation setting, $\alpha_i$ is shared by the two activations $\varphi_{i,j}$ but not across the channels. This guarantees that the activation used across channels are rescaled versions one of another.
%\end{remark}
\begin{remark}
\label{rk:weakboundscaling}
    The bound on the weak-convexity modulus given in Proposition \ref{pr:weakcvxconstantBound} does not depend on the parameters $\alpha_i$. Consequently, the addition of the $\alpha_i$ as learnable parameters does not compromise the weak-convexity guarantees on $R$. 
\end{remark}
In the remainder of the paper, $\boldsymbol{\theta}$ represents the aggregated set of learnable parameters (as detailed in Section \ref{sec:TrainingProc}) and we use the notation $R_{\boldsymbol{\theta}}$ whenever an explicit reference to the parameters is needed.
Likewise, with a slight abuse of notation, we use $R_{\boldsymbol{\theta}(\sigma)}$ to denote the regularizer at noise level $\sigma$.
This noise-dependent regularizer $R_{\param(\sigma)}$ then yields the proximal denoiser
\begin{equation}\label{eq:denoiser}
D_{\param(\sigma)}(\vec y) = \prox_{R_{\param(\sigma)}}(\vec y) = \argmin\limits_{\vec{x}\in\R^d} \frac{1}{2}\|\vec{x} - \vec{y}\|_2^2 + R_{\param(\sigma)}(\vec{x}).
\end{equation}
In general, $D_{\param(\sigma)}$ does not have a closed-form expression but, due to the convexity and smoothness of the underlying objective, $D_{\param(\sigma)}(\vec y)$ can be computed efficiently with gradient-based solvers.
In practice, we use AGD \cite{Nesterov1983} combined with the standard gradient-based restart technique introduced in \cite{ODonoghue2015}.
The stepsize is chosen as $1/(1 + \max(1, \mu))$, which ensures convergence to a global minimizer as a consequence of the Lipschitz bound \eqref{eq:estLip}.
\subsection{Training Procedure}\label{sec:TrainingProc}
In this section, we detail how the parameters $\boldsymbol{\theta}$ are learned so that $D_{\param(\sigma)}(\vec y)$ is a good Gaussian denoiser across multiple noise levels.
\subsubsection{Training Problem}Let $\{\vec x^m\}_{m=1}^M$ be a set of clean images.
Each image $\vec x^m$ is corrupted as $\vec y^m = \vec x^m + \sigma^m \vec n^m$ with Gaussian noise $\vec n^m\sim \mathcal N(\vec 0,\mat I)$ and a noise level $\sigma^m \sim \mathcal U[0,\sigma_{\max}]$.
Then, we define the following multi-noise-level training problem
\begin{equation}\label{eq:TrainProb}
    \hat{\param}\in \argmin\limits_{\param}\sum_{m=1}^M \mathbb E_{(\vec n^m, \sigma^m)} \bigl(\Vert D_{\param(\sigma^m)}(\vec y^m) - \vec x^m\Vert_1\bigr).
\end{equation}
Here, the $\ell_1$ loss is chosen because it is known to be robust and well-performing for the training of CNNs \cite{LossZaho2017, KKHP2017}.
\subsubsection{Optimization}
\label{subsec:optimisation}
For clarity, we briefly recall the various parameters contained in $\boldsymbol{\theta}$ before outlining the actual optimization procedure.
\paragraph{Profile-Related Parameters}
The linear splines $\varphi_+$ and $\varphi_-$ are parameterized by $\vec c_+$ and $\vec c_-$ via the constrained coefficients $\tilde{\vec c}_\pm = \frac{1}{2}(\boldsymbol{P}_{\uparrow}(\vec c_\pm)-\texttt{reverse}(\boldsymbol{P}_{\uparrow}(\vec c_\pm)))$ so that they are odd, non-decreasing, and non-expansive. Together with $\mu>0$, \eqref{eq:decomp_profile} then leads to the linear-spline activation function $\varphi$.
Recall that its primitive defines the profile $\psi$.
The parameters $\vec c_{\alpha}^i$ specify the linear-spline functions $\alpha_i(\sigma)$, which rescale the profile $\psi$ across the channels in \eqref{eq:convexridgereg} and across the noise levels.
\paragraph{Spectral Normalization}
The convolution operation represented by $\mat W$ is parameterized as $\mat W=\mat U/\|\mat U\|$, $\mat U$ consisting in the composition of $3$ zero-padded convolutions.
Here, $\|\mat U\|$ is computed as follows.
\begin{itemize}
    \item {\textbf{Training stage}:} Assuming circular boundary conditions instead of zero-padding, we can consider $\mat U$ as a \textit{single-convolution} layer.
    Then, $\mat U^T \mat U$ encodes a $2$D convolution from a one-channel input to a one-channel output.
    Hence, it can be represented by a kernel $\mat K_{\mat U^T \mat U} \in \R^{(2K_s -1)\times (2K_s-1)}$.
    In this setting, the spectrum of $\mat U^T \mat U$ can be computed using the $2$D discrete Fourier transform (DFT) \cite{sedghi2018the} as
    \begin{equation}
    \label{eq:convspectrum}
    \mathrm{spec}(\mat U^T \mat U) = \bigl\{|\mathrm{DFT}(\mathrm{Pad}_{\sqrt{d}}(K_{\mat U^T \mat U}))_{k_1 k_2}|\colon 1\leq k_1,k_2\leq \sqrt{d}\bigr\},
    \end{equation}
    where $\mathrm{Pad}_{\sqrt{d}}$ zero-pads $\mat K_{\mat U^T \mat U}$ into a $(\sqrt{d}\times \sqrt{d})$ image.
    We rely on \eqref{eq:convspectrum} to estimate $\|\mat U\|\simeq \max (\mathrm{spec}(\mat U^T \mat U))$ during training since it is efficient to compute and can be incorporated into the computational graph of the computation of $\boldsymbol{\nabla}R$.
    \item {\textbf{Test stage}:} Subsequently, when evaluating a trained WCRR-NN model, the true $\|\mat U\|$ is computed with high precision using the power method (1000 steps).
This firm normalization guarantees the $1$-weak convexity of the underlying $R$ (up to numerical imprecision).
\end{itemize}
\paragraph{Implicit-Differentiation}
The learning of the proximal denoiser comes with the challenge that $D_{\param(\sigma)}$ depends implicitly on $\boldsymbol{\theta}$.
As shown in the deep-equilibrium (DEQ) framework \cite{BKK2019}, it is possible to compute the Jacobian $J_{\param} D_{\param(\sigma)}$ of the denoiser with respect to the parameters via implicit differentiation.
For this purpose, two steps are required.
\begin{itemize}
    \item {\textbf{Image denoising}:} First, given a noisy input $\vec y^m$, one needs to perform the forward pass, which consists in the computation of $\hat{\vec x} = D_{\param(\sigma^m)}(\vec y^m)$.
    The deployed AGD is run until the relative change of norm between consecutive iterates is lower than $\num{e-4}$.
\item {\textbf{Gradient computation}:} We use the DEQ implementation introduced in \cite{BKK2019} and now briefly discuss the general concept within our setting.
The differentiability of $R$ implies that the denoised images satisfy
\begin{equation}\label{eq:optimality}
   \hat{\vec x}(\boldsymbol{\theta}) - \vec y + \boldsymbol{\nabla}_{\vec x} R(\boldsymbol{\theta}, \hat{\vec x}(\boldsymbol{\theta})) = \vec 0,
\end{equation}
where the dependence on $\sigma$ is dropped for clarity and the dependence on $\boldsymbol{\theta}$ is made explicit.
The application of the implicit-function theorem for \eqref{eq:optimality} leads to
\begin{equation}
   (\mat I + \mat H_R(\boldsymbol{\theta}, \hat{\vec x}(\boldsymbol{\theta}))) J_{\param} \hat{\vec x}(\boldsymbol{\theta})  = J_{\param} (\boldsymbol{\nabla}_{\vec x} R) (\boldsymbol{\theta}, \hat{\vec x}(\boldsymbol{\theta})).
\end{equation}
Hence, we evaluate the matrix-vector products with $(J_{\param} D_{\param}(\vec y))^T = (J_{\param} \hat{\vec x}(\boldsymbol{\theta}))^T$ (which are required for computing the gradients of \eqref{eq:TrainProb} within the backpropagation algorithm) by solving a simple linear system.
This is carried out with the Anderson routine given in \cite{BKK2019}.
While deriving $J_{\param} (\boldsymbol{\nabla}_{\vec x} R)$ is cumbersome and usually left to automatic differentiation, we use the explicit expression
\begin{equation}
    \mat H_R(\hat{\vec x})\vec u = \mat W^T \bigl(\boldsymbol{\spline}'(\mat W {\hat{\vec x}}) \odot (\mat W {\vec u})\bigr),
\end{equation}
where $\odot$ is the Hadamard product and the piecewise-constant function $\boldsymbol{\spline}'$ is analytically derived from the B-spline representation of the linear spline $\spline$.
This yields the same results as automatic differentiation, but was found to be more efficient.
\end{itemize} 

\paragraph{Optimization}
The non-convex training problem in \eqref{eq:TrainProb} is solved with the stochastic Adam optimizer \cite{KinJim2015}, where we sample for each batch the $\vec x^m$, the corresponding noise-level $\sigma^m$, and the noise $\vec n^m$.
Note that, within each batch, images are corrupted with different noise levels and, likewise, in different epochs, different noise levels can be applied to the same $\vec x^m$.

\subsection{Training and Denoising Performance}\label{sec:WeakConvReg}
The proposed weakly convex regularizer is learned over the Gaussian-denoising task described in Section~\ref{sec:TrainingProc}, with $\sigma_{\max} = 30/255$.
The same procedure as in \cite{bohra2021learning} is used to form 238,400 grayscale patches\footnote{WCRR-NNs are fully convolutional and can process input of any spatial size.} of size $(40\times 40)$ from 400 images of the BSD500 data set \cite{arbelaez_contour_2011}, while 12 other images are kept for validation.
In accordance with the ablation study reported in Tables \ref{table:denoising_performance_channels} and \ref{table:denoising_performance_kernels}, the three filters in $\mathbf{U}$ have kernels with $k_s = 5$ and $4$, $8$, and $60$ output channels, respectively.
The linear splines $\spline_i$ have $M+1=101$ equally distant knots with $\Delta = \num{2e-3}$.
We initially set $\vec c_+ =\vec 0$ and $(\vec c_-)_m=\tau_m$, which was found to be important to help the training. Intuitively, this choice helps the regularizer to use weak convexity, which is only permitted through $\vec c_{-}$.
The linear splines parameterizing $\alpha_i(\sigma)$ have $11$ equally distant knots in the range $[0,\sigma_{\max}]$ and are initialized with the constant value $5$.
Our model is trained with the Adam optimizer for $6000$ steps with batches of size $128$, which takes less than 2 hours on a Tesla V100 GPU.
The learning rates are initially set to $\num{5e-2}$ for $\mu$, $\num{5e-3}$ for $\mat U$, and $\vec c_{\alpha}^i$, and to $\num{5e-4}$ for $\vec c_+$ and $\vec c_-$.
Then, they are decayed by 0.75 every $500$ batches.
For evaluation, the denoising \eqref{eq:denoiser} is performed with AGD and a tolerance of $\num{e-4}$ for the relative change of norm between consecutive iterates.
An example of convergence curves is provided in Figure \ref{fig:convergencedenoising}.

\begin{figure}[!t]
    \centering
    \includegraphics[width=0.9\textwidth]{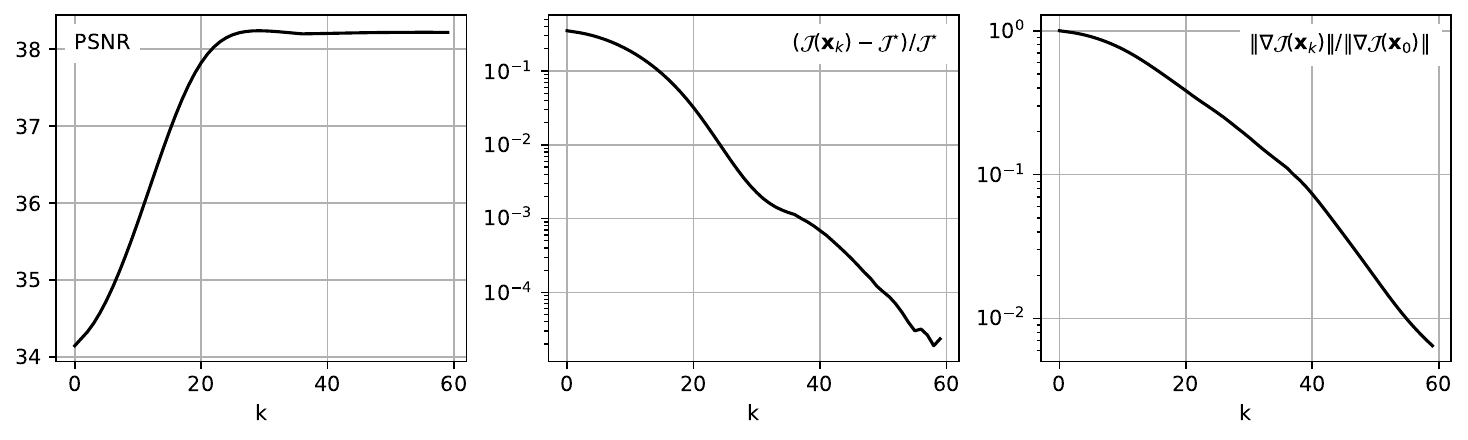}
    \caption{Example of convergence curves for denoising with the WCRR-NN and AGD.}
    \label{fig:convergencedenoising}
\end{figure}

\label{subsec:denoisingcomparison}
The numerical evaluation of our WCRR-NNs and several other methods on the BSD68 test set is provided in Table~\ref{table:denoising_performance}.
The task is non-blind, in the sense that the noise level is used either directly as an input (as in\ BM3D) or indirectly via a regularization parameter that is tuned on a corresponding validation set (as in\ TV).
The first important observation is that WCRR-NNs, which implement a convex energy, outperform the popular BM3D denoiser~\cite{DabFoiKat2007}.
To the best of our knowledge, this is the first time a (learnable) convex model surpasses BM3D.
A visual comparison of BM3D and (W)CRR-NNs is provided in Figure~\ref{fig:DenoiseCompare}.
The results obtained with the 2nd and 6th methods in Table~\ref{table:denoising_performance} on the same image can be found in the original paper \cite{CheRan2014}.
Next, we discuss in more depth the  frameworks from Table~\ref{table:denoising_performance} that are close in spirit to WCRR-NNs.
\begin{table}[t]
\caption{Denoising performance on the BSD68 test set.}
\label{table:denoising_performance}
\setlength\tabcolsep{2pt}
\centering
\begin{tabular}{llccc}
\toprule
 & & $\sigma=5/255$ & $\sigma=15/255$ & $\sigma=25/255$ \\
\midrule
\multirow{3}{*}{Convex} & TV$^{1,2,3}$\cite{rudin1992nonlinear} & 36.41 & 29.90 & 27.48 \\
& Higher-order MRFs convex$^{1,2,3}$ \cite{CheRan2014} & - & 30.45 & 28.04\\
& CRR-NN$^{1,2,3}$\cite{GouNeuBoh2022} & 36.96 & 30.55 & 28.11\\
\midrule
\multirow{2}{*}{Provably CNC} & TV CNC$^{1,2}$ & 36.53 & 29.92 & 27.49 \\
& WCRR-NN$^{1,2}$ & 37.68 & 31.22 & 28.69\\
\midrule
Approx.\ CNC & Prox-DRUNet & 37.98 & 31.70 & 29.18 \\
\midrule
\multirow{2}{*}{Others} & Higher-order MRFs$^{1}$ \cite{CheRan2014} & - & 31.22 &  28.70\\
& BM3D \cite{DabFoiKat2007} & 37.54 & 31.11 & 28.60\\
\bottomrule
\multicolumn{4}{l}{\footnotesize $^1$Ridge-based regularizer, \footnotesize $^2$Minimization of convex functional, \footnotesize $^3$Convex regularizer}\\
\end{tabular}
\end{table}
\begin{figure}[t]
    \centering\includegraphics[width=\linewidth]{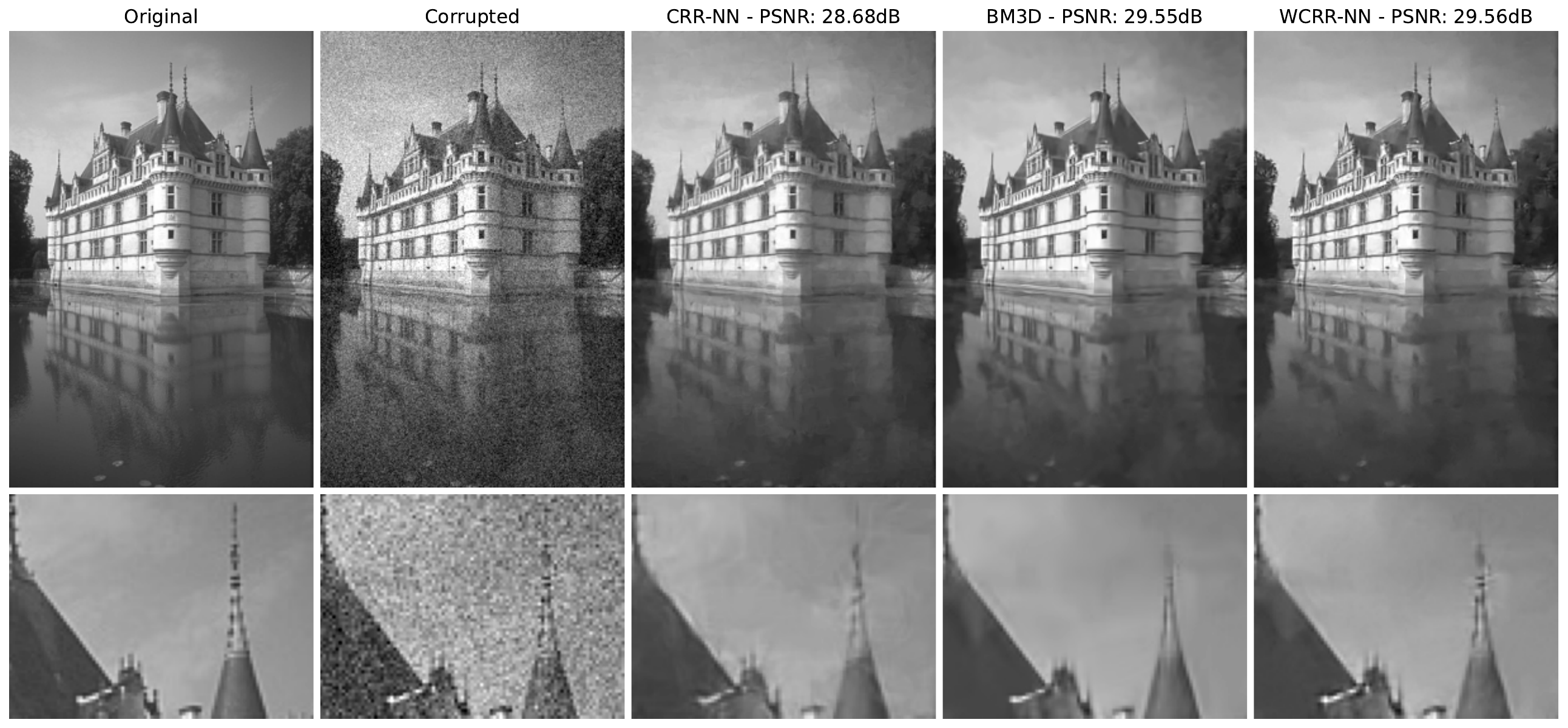}
    \caption{Denoising of the ``castle" image from the BSD68 test set for noise level $\sigma=25$.}
    \label{fig:DenoiseCompare}
\end{figure}

\begin{table}[t]
\parbox{.57\linewidth}{\caption{WCRR-NN: PSNR on BSD68 vs number of filters.}
\label{table:denoising_performance_channels}
\setlength\tabcolsep{3pt}
\centering
\begin{tabular}{lccccc}
\toprule
 $N_c$ & 10 & 20 & 40 & 60 & 80 \\
 \midrule
 $\sigma=5/255$ & 37.43 & 37.59 & 37.63 & 37.66 & 37.66 \\
 $\sigma=15/255$ & 30.85 & 31.15 & 31.19 & 31.21 & 31.20 \\
 $\sigma=25/255$ & 28.16 & 28.62 & 28.67 & 28.68 & 28.69\\
\bottomrule
\end{tabular}}
\hfill
%\begin{tabular}{lccc}
%\toprule
% \#channels & $\sigma=5/255$ & $\sigma=15/255$ & $\sigma=25/255$ \\
%\midrule
%10 & 37.43 & 30.85 & 28.16 \\
%20 & 37.59 & 31.15 & 28.62 \\
%40 & 37.63 & 31.19 & 28.67 \\
%60 & 37.66 & 31.21 & 28.68 \\
%80 & 37.66 & 31.20 & 28.69 \\
%\bottomrule
%\end{tabular}
\parbox{.42\linewidth}{\caption{WCRR-NN: PSNR on BSD68 vs kernel size.}
\label{table:denoising_performance_kernels}
\centering
\setlength\tabcolsep{3pt}
\begin{tabular}{lccc}
\toprule
 $k_s$ & 3 & 5 & 7 \\
 $K_s$ & 5 & 13 & 19 \\
 \midrule
 $\sigma=5/255$ & 37.64 & 37.66 & 37.65 \\
 $\sigma=15/255$ & 31.14 & 31.21 & 31.21 \\
 $\sigma=25/255$ & 28.56 & 28.68 & 28.68\\
\bottomrule
\end{tabular}}
\end{table}

\paragraph{CNC-Based Total Variation}
The WCRR-NN model is inspired by earlier works that extend TV denoising to the CNC framework using non-convex potential functions \cite{ZSZLLD2019, ScrChoPes2022}. 
The publicly available implementations outperform TV for specific classes of images, typically for cartoon-like ones with sharp edges.
However, we did not observe any significant improvements for the denoising of the natural images in BSD68.
Hence, in our comparison, we used our own version of CNC-TV, which was obtained by training a WCRR-NN with two fixed filters, namely, the horizontal and vertical finite differences.
This corresponds to an anisotropic TV denoising with learned profiles---the CNC counterpart of the standard anisotropic TV denoising model.
As reported in Table~\ref{table:denoising_performance}, this only yields marginal improvements over TV.
\paragraph{Field of Experts and Higher-Order MRFs}
The FoE approach corresponds to learning the filters associated with a regularizer of the form \eqref{eq:convexridgereg} with hand-picked profile functions \cite{RotBla2009}.
It was successfully applied in \cite{CheRan2014}, with both convex and non-convex profiles.
A key difference with WCRNNs lies in the theoretical guarantees: The non-convex profiles are unconstrained in \cite{RotBla2009, CheRan2014}.
Hence, the objective function is not provably convex.
This means that the optimization is delicate and the convergence to a global optimum cannot be guaranteed.
Interestingly, WCRR-NN offer the same performance as in \cite{CheRan2014} while minimizing a convex energy.

\paragraph{CRR-NNs}
Our work extends upon the convex regularizers learned with CRR-NNs~\cite{GouNeuBoh2022}.
The substitution of weak convexity for convexity makes a significant difference as it yields a gain of at least 0.6dB for all noise levels (see Table~\ref{table:denoising_performance}).
In contrast with the simpler 2-filter TV setting, the improvement is substantial.
This indicates that the learning of sufficiently many filters is necessary to fully exploit the additional freedom provided by weak convexity.

\paragraph{Gradient-Step Denoisers}
Proposition~\ref{prop:ProxCNC} allows one to (implicitly) construct non-convex regularizers $R$ by learning their proximal operator.
This result is exploited in \cite{HurLec2022}, where $\prox_{R}=\boldsymbol{\nabla} \psi$ is parameterized through the potential $\psi= \frac{1}{2}\Vert \cdot \Vert^2 + g$, where $\boldsymbol{\nabla} g$ must be contractive.
To leverage the power of deep-learning, the authors choose $g=\frac{1}{2}\|\cdot - \mathrm{DRUNet}(\cdot)\|_2^2$, where DRUNet \cite{Drunet2022} is a deep CNN with $\sim17$ million parameters.
As there are currently no efficient methods to globally bound the Lipschitz constant of the gradient of a deep CNN, they propose to instead regularize the norm of the Jacobian of $\boldsymbol{\nabla} g$ at finitely many locations during training.
This yields the Prox-DRUNet denoiser, which performs very well in practice (see Table \ref{table:denoising_performance}\footnote{The Prox-DRUNet denoiser given in \cite{HurLec2022} is trained on color images.
For grayscale denoising, we plug the image into all three color channels and average the output across the channels, and tune the denoising strength parameter $\sigma$ to optimize performance.
As expected, the obtained metrics are on par with DnCNN for $\sigma=25$ and with the gradient-step denoiser for $\sigma=5$ in \cite{hurault2022gradient}, which indicates the appropriateness of the usage.}). Note that Prox-DRUNet only approximately satisfies the conditions to be a truly CNC method because $\|\mat H_g(\vec x)\|$ can be greater than 1 for some $\vec x$, meaning that $\boldsymbol{\nabla}g$ is not contractive (as already reported in \cite{HurLec2022}).
On noisy BSD68 images, we found\footnote{We computed $\|\mat H_g\|$ with a precise power method (300 iterations).} that $\|\mat H_g\|$ can be as large as 1.07 ($\sigma=5$), 1.08 ($\sigma=15$), 1.18 ($\sigma=25$), and on a set of 68 random images (i.i.d. uniformly distributed pixels in $[0, 1]$) as large as 1.69 ($\sigma=5$), 1.20 ($\sigma=15$), 1.43 ($\sigma=25$).
%Similar findings are reported in \cite{HurLec2022}.
Overall, we believe that WCRR-NNs and Prox-DRUNet offer a very complementary perspective. In fact, the good performance of Prox-DRUNet suggests that there could even be some room for further improvements with provably CNC methods.

\subsection{Interpretation as Sparsity Prior}\label{sec:UnderHood}
The filters and profile functions learned for our WCRR-NNs are shown in Figures~\ref{fig:Filters} and~\ref{fig:Profiles}, respectively.
\paragraph{Filters}
The impulse responses of the filters in $\mat W$ present patterns akin to wavelets and Gabor filters, in that they come in various modulations, orientations, and scales.
In addition, the kernel $\mat K$ corresponding to the convolution $\mat W^T \mat W$ is very close to the $2$D discrete Kronecker impulse, meaning that $\mat W$ is almost a Parseval frame ($\mat W^T \mat W\simeq \mat I$).
A key difference, however, is that $\mat K$ is zero-mean.
We also observed that more filters than in the convex setting of CRR-NNs are needed to reach the maximal performance.
The payoff is that the filters are now able to capture more complicated patterns.
\paragraph{Profile Functions}
The learned profiles $\psi_i$ are shared among the filters and then individually rescaled with the $\alpha_i$, so that the $\psi_i$ have the same shape. 
Hence, only their prototype $\psi$ is discussed here.
The latter converges to a quasi-convex function (i.e., sub-level sets are intervals) even without us explicitly imposing this constraint. 
Moreover, $\psi$ fully exploits the 1-weak convexity of the regularizer $R$ in the sense that $\min_{t} \psi^{\prime \prime} (t)= -1$.
Hence, this is an active constraint since $R$ would not satisfy it by default. Overall, $\psi$ closely resembles the minimax concave penalty function \cite{LanMorSel2019}.
%As observed with CRR-NNs \cite{GouNeuBoh2022}, the derivative $\spline=\psi^{\prime}$ is very steep around zero, which makes $\psi$ almost non-differentiable at zero.

To extend our model, we also experimented with learning a different $\psi_i$ for each filter.
This led to less interpretable profiles (not necessarily quasi-convex and with some oscillations), while it only offered a negligible gain in performance: less than 0.05dB on the denoising experiment for noise levels $\sigma\in\{5/255,15/255,25/255\}$.
\begin{figure}[tbp]
    \parbox{.61\linewidth}{\centering{\includegraphics[width=\linewidth]{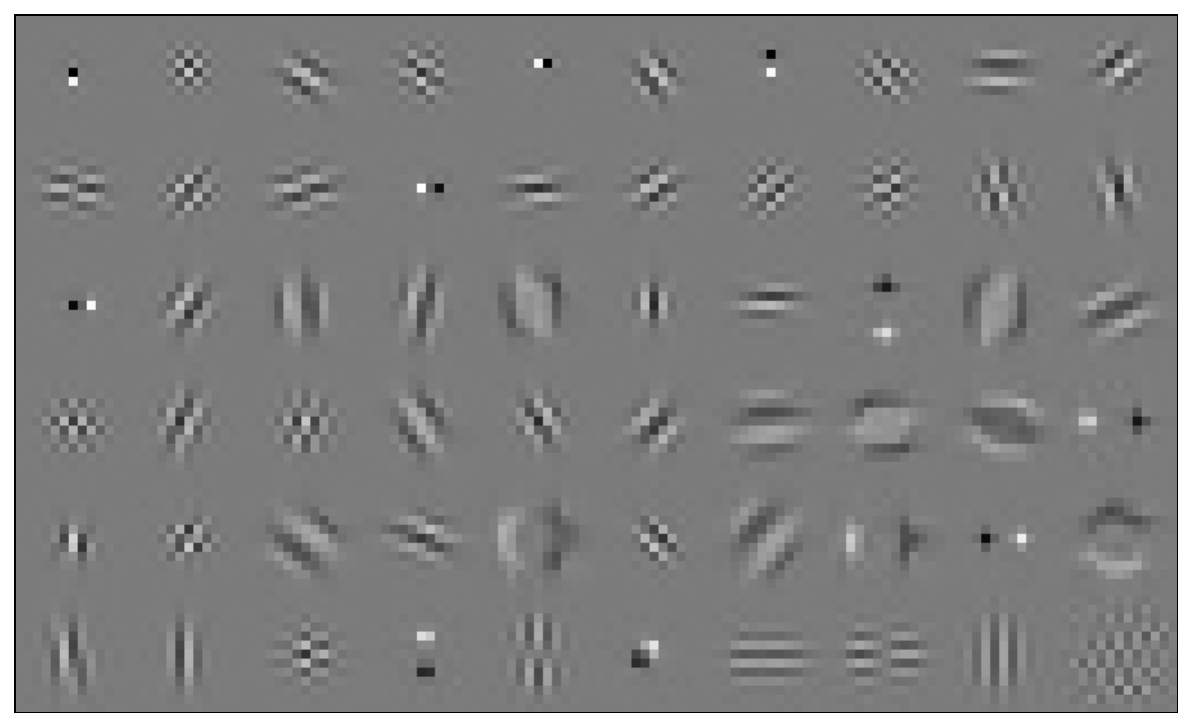}}
    \caption{Impulse response of the filters in the learned WCRR-NN.}
    \label{fig:Filters}}
    {\parbox{.38\linewidth}{\centering{\includegraphics[width=0.98\linewidth]{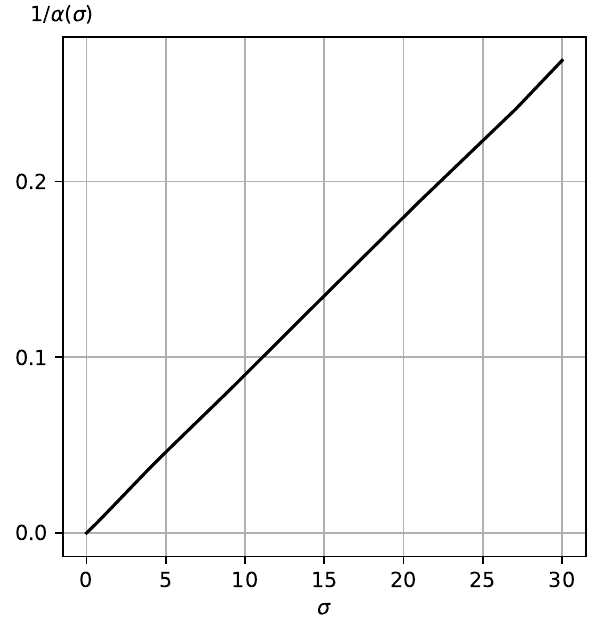}}
    \caption{Plot of $1/\alpha(\sigma)$ vs $\sigma$, where $\alpha(\sigma)=\sum_{i=1}^{N_C} \alpha_i(\sigma)$ encodes the average behavior across the channels.}
    \label{fig:Scaling}}}
    \vspace{.3cm}

    \centering{\includegraphics[width=145mm]{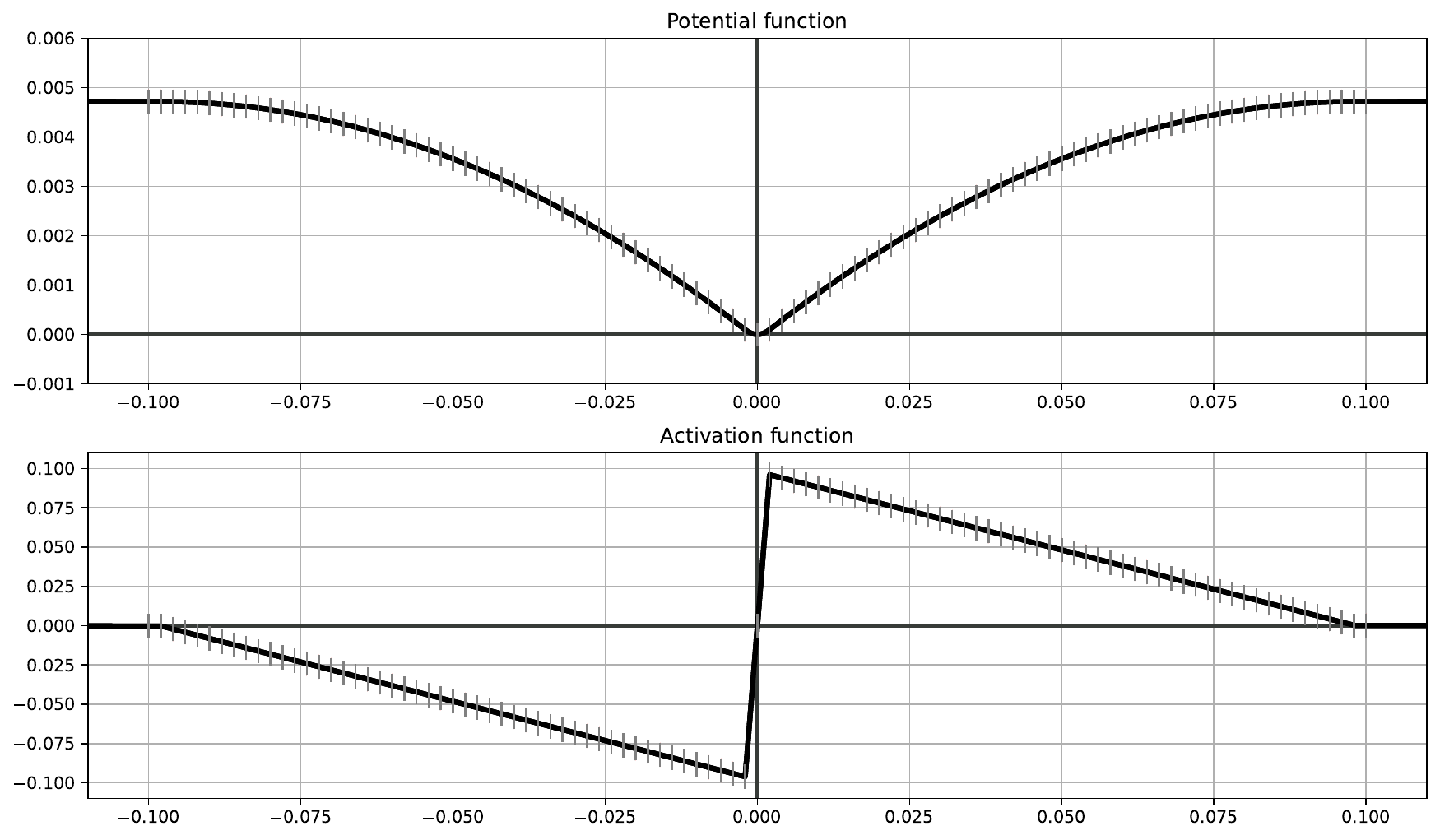}}
    \caption{Potential function $\psi$ and activation function $\varphi=\psi'$ of the learned WCRR-NN.
    These functions are splines of degrees 2 and 1, respectively.
    The vertical markers indicate the control points of the splines.}
    \label{fig:Profiles}
\end{figure}

\paragraph{Noise-Dependent Scaling} The only part of $R$ that depends on the noise level $\sigma$ is the profiles $\psi_i$, which depend on $\sigma$ through the $\alpha_i(\sigma)$.
As can be seen in Figure \ref{fig:Scaling},
the $1/\alpha_i$ are on average linear functions of $\sigma$.
Loosely speaking, most profiles will roughly have the form $\psi_i(t, \sigma) \simeq \sigma^2\psi(t/\sigma)$.
To verify that such a simple dependence of $R$ on $\sigma$ is sufficient, 3 WCRR-NNs were trained to denoise at a single noise level ($\sigma\in\{5/255,15/255,25/255\}$).
As these models do not outperform the multi-noise-level WCRR-NN on BSD68, the simple rescaling of the profiles appears to suffice.
\paragraph{Signal-Processing Perspective}
The regularizer $R$ is trained to promote natural images. The corresponding gradient-descent step\footnote{In our setting with no biases and where $\spline$ has a maximum slope at the origin, it can be shown that $\mathrm{Lip}(\mat W^T \boldsymbol{\spline}(\mat W \cdot)) = \|\spline'\|_{\infty}$.} $\vec x \mapsto \vec x - \boldsymbol{\nabla}R(\vec x)/\mathrm{Lip}(\boldsymbol{\nabla}R) = \vec x - \mat W^T \boldsymbol{\spline}(\mat W \vec x)/\|\varphi'\|_{\infty}$, which should increase the regularity of images, is therefore expected to remove features considered as noise in natural images.
In turn, we then expect that $\vec x \mapsto \mat W^T \boldsymbol{\spline}(\mat W \vec x)/\|\varphi'\|_{\infty}$ extracts some noise.
Due to its shape, see Figure \ref{fig:Profiles}, the function $\spline/\|\spline'\|_{\infty}$ preserves the small responses $\mat W \vec x$ to the filters (it is almost the identity for small inputs), and cuts the large ones (it is almost the zero function for large inputs).
Hence, one reconstructs the estimated noise $\mat W^T \boldsymbol{\spline}(\mat W \vec x)/\|\varphi'\|_{\infty}$ by essentially removing the components of $\vec x$ that exhibit a significant correlation with the kernels.
This allows for a more efficient noise extraction than done by the monotonic clipping function learned in the convex regularization framework of CRR-NNs, see \cite[Figs.~5 and 6]{GouNeuBoh2022}.
While the monotonic clipping also preserves the small inputs, it is unable to fully remove the large responses because of the monotonicity constraint stemming from the convexity of the underlying potential. 

In addition to the above perspective, we can make a link with wavelet- or framelet-like denoising \cite{Donoho1995,LangGOBW1996,ChBiVet2000,BluLui2007,ParekhASW2015}.
Indeed, given that $\mat W^T\mat W \simeq \mat I$, the gradient-descent step can be approximated as $\vec x \mapsto \vec x - \mat W^T \boldsymbol{\spline}(\mat W \vec x)/\|\varphi'\|_{\infty} \simeq \mat W^T \phi(\mat W \vec x)$ with $\phi = \mathrm{Id} - \spline/\|\spline'\|_{\infty}$.
Since $\phi$ is zero around the origin and is the identity for sufficiently large inputs, it qualitatively stands between the soft- and hard-thresholding functions that have been key components for wavelet and framelet denoising for years.
Finally, note that framelet-denoising models are themselves closely related to proximal operators \cite{HHNP19}.

\section{Extension to Generic Inverse Problems}
\label{sec:InverseProblems}
We now use the regularizer $R_{\param(\sigma)}$ trained in Section~\ref{sec:Denoising} to solve inverse problems based on the variational formulation \eqref{eq:VarProb}.
Here, the key challenge is the possible non-convexity of the objective, which prevents us from minimizing \eqref{eq:VarProb} globally.
It is, however, possible to search for critical points.
These are still of particular interest, especially because the regularizer has a simple structure with an \emph{almost} convex energy landscape.

\paragraph{Proximal vs Gradient Methods}
The standard PnP frameworks rely on proximal-based methods with an explicit denoising step. The motivation there is that the denoising step is typically efficient to perform while neither the regularizer (if it exists) is explicitly known, nor its gradient. In our setting, on the contrary, it is very efficient to evaluate the regularizer and its gradient, and hence AGD methods \cite{Nesterov1983}, which are applicable to general non-convex problems \cite{GhaLan2016}, are better suited. In our setting, AGD is also known to attain optimal convergence rates among first-order methods.
In the sequel, we recall the main features of AGD and show how to leverage the knowledge of the weak-convexity modulus of the objective.

% In our framework is also possible, but we found this to be less efficient in practice.
\subsection{Accelerated Gradient Descent}
To solve the inverse problem, we minimize the regularized objective
\begin{equation}\label{eq:InvProbLearn}
    \mathcal J(\vec x) = \frac{1}{2}\Vert \vec H\vec x - \vec y \Vert_2^2 + \lambda R_{\param(\sigma)}(\vec x),
\end{equation}
where $R_{\param(\sigma)}$ is the $1$-weakly convex regularizer from Section~\ref{sec:WeakConvReg} and $\lambda>0$ is a regularization parameter.
Since the objective is differentiable, we can rely on gradient-based methods to find critical points of $\mathcal J$ as a convenient alternative to proximal algorithms.
To reduce the reconstruction time, we propose an AGD variant in Algorithm~\ref{alg:AGD}, which is tailored to \mbox{$\lambda$-weakly} convex functionals $\mathcal J$ with $L$-Lipschitz-continuous gradient.

From \eqref{eq:estLip}, we infer that $\boldsymbol{\nabla} \mathcal J$ is $L$-Lipschitz-continuous with $L\leq \Vert\mat  H \Vert^2 + \lambda \max(\mu,1)$, which implies for $\vec x_1, \vec x_2 \in \R^d$ the standard upper estimate
\begin{align}
    \mathcal J (\vec x_1) \leq \mathcal J (\vec x_2) + \boldsymbol{\nabla} \mathcal J (\vec x_2)^T (\vec x_1 - \vec x_2) + \frac{L}{2} \Vert \vec x_1 - \vec x_2 \Vert^2\label{eq:QuadEst}.
\end{align}
As $R_{\param(\sigma)}$ is 1-weakly convex, $\mathcal J$ is $\lambda$-weakly convex.
Hence, the subgradient inequality for convex functions leads for $\vec x_1, \vec x_2 \in \R^d$ to the lower estimate
\begin{align}
    \mathcal J (\vec x_1) \geq \mathcal J (\vec x_2) + \boldsymbol{\nabla} \mathcal J (\vec x_2)^T (\vec x_1 - \vec x_2) - \frac{\lambda}{2} \Vert \vec x_1 - \vec x_2\label{eq:WeakConvEst}\Vert^2.
\end{align}
\begin{algorithm}[t]
	\begin{algorithmic}
	    \State \textbf{Input:} initialization $\vec \vec x_{0} \in \R^d$, tolerance $\epsilon>0$, $a>1$
		\State Set $t_0 = t_{1} = 1$, $k=1$, $\vec z_0 = \vec x_0$, $\vec x_1 =  \vec x_0$
		\While{$\Vert \vec x_k - \vec x_{k-1}\Vert/\Vert \vec x_{k-1}\Vert > \epsilon$ or $k=1$}
			\State $\vec z_{k} = \vec x_k + \frac{t_{k-1} - 1}{t_k} (\vec x_k - \vec x_{k-1})$
            \State $\text{crit} = \boldsymbol{\nabla} \mathcal J (\vec z_{k})^T(\vec z_{k} - \vec z_{k-1}) + \frac{a\lambda}{2} \Vert \vec z_{k} - \vec z_{k-1} \Vert^2$
			\If{$\text{crit} > 0$}
			    \State $\vec z_{k} = \vec x_k$
                \State $t_k=1$
			\EndIf
			\State $\vec x_{k+1} = \vec z_k - \frac{1}{L} \boldsymbol{\nabla} \mathcal J(\vec z_k)$
			\State $t_{k+1} = \frac{1+\sqrt{1+4 t_k^2}}{2}$
            \State $k \gets k+1$
		\EndWhile
		\State \textbf{Output:} Approximate solution $\vec x_k$
		\caption{Safeguarded AGD for $\lambda$-weakly convex $\mathcal J$ with $L$-Lipschitz gradient}
		\label{alg:AGD}
	\end{algorithmic}
\end{algorithm}
Given some initialization $\vec x_0 = \vec x_{-1} =\vec z_0 \in \R^d$ and a sequence of Nesterov momentum parameters $\{\beta_k\}_{k \in \N}\subset [0,1]$, the standard AGD \cite{Nesterov1983} update steps read
\begin{align}
    \vec z_{k} &= \vec x_k + \beta_k (\vec x_k - \vec x_{k-1}),\label{eq:AGD_Exp}\\
    \vec x_{k+1} &= \vec z_k - \frac{1}{L} \boldsymbol{\nabla} \mathcal J(\vec z_k).\label{eq:AGD_Up}
\end{align}
The combination of \eqref{eq:AGD_Up} and \eqref{eq:QuadEst} yields the decrease estimate
\begin{equation}\label{eq:DecreaseEst}
    \mathcal J(\vec x_{k+1}) - \mathcal J(\vec z_{k}) \leq - \frac{L}{2} \Vert \vec x_{k+1} - \vec z_{k} \Vert^2.
\end{equation}
However, the update \eqref{eq:AGD_Exp} does not necessarily guarantee the decrease of $\{\mathcal J(\vec z_{k})\}_{k \in \N}$.
Hence, for a predefined $a>1$, we propose to check the condition
\begin{equation}\label{eq:CondCheck}
    \boldsymbol{\nabla} \mathcal J (\vec z_{k})^T(\vec z_{k} - \vec z_{k-1}) + \frac{a\lambda}{2} \Vert \vec z_{k} - \vec z_{k-1} \Vert^2 \leq 0
\end{equation}
after having tentatively performed \eqref{eq:AGD_Exp}.
If \eqref{eq:CondCheck} is violated, we perform the plain gradient update $\vec z_{k} = \vec x_{k}$ (instead of \eqref{eq:AGD_Exp}) and apply a restart technique, as proposed in \cite{ODonoghue2015}.
In this case, we get from \eqref{eq:DecreaseEst} at the previous step that
\begin{equation}\label{eq:EstVio}
    \mathcal J(\vec z_{k}) - \mathcal J(\vec z_{k-1}) \leq - \frac{L}{2} \Vert \vec z_{k} - \vec z_{k-1} \Vert^2.
\end{equation}
Otherwise, the incorporation of \eqref{eq:WeakConvEst} implies that
\begin{align}\label{eq:EstNonVio}
    \mathcal J(\vec z_{k}) - \mathcal J(\vec z_{k-1}) \leq \boldsymbol{\nabla} \mathcal J (\vec z_{k})^T (\vec z_{k} - \vec z_{k-1}) + \frac{\lambda}{2} \Vert \vec z_{k} - \vec z_{k-1}\Vert^2 \leq  -\frac{(a - 1)\lambda}{2} \Vert \vec z_{k} - \vec z_{k-1} \Vert^2.
\end{align}
To sum up, the acceleration steps are kept only if they lead to a sufficient decrease of the objective.
Otherwise, the plain gradient-descent step guarantees this decrease.
\begin{remark}
    The gradient-based condition \eqref{eq:CondCheck} is more restrictive than the objective-based condition \eqref{eq:EstVio}. 
    However, \eqref{eq:CondCheck} is computationally cheaper to verify as it only involves inner products of already computed quantities.
    In practice, we observed that \eqref{eq:CondCheck} is rarely violated.
\end{remark}
\begin{remark}
    The parameter $a$ in \eqref{eq:CondCheck} must be greater than the weak-convexity modulus of $R_{\param}$ to ensure convergence.
    At this point, a precise estimate of this modulus---which we know to be bounded by one in our setting---weakens the condition \eqref{eq:CondCheck} and typically yields faster convergence.
    On the contrary, the reliance on a loose bound leads to frequent restarts, at the detriment of acceleration.
\end{remark}
Regarding Algorithm~\ref{alg:AGD}, we now derive a convergence result in Theorem~\ref{thm:ConvGD} using \cite[Theorem ~3.7]{OchCheBro2014}, which itself extends the seminal work \cite{AttBolSva2013} to the inertia setting.
Note that the objective \eqref{eq:InvProbLearn} is semi-algebraic since the profile function $\psi$ is piecewise-polynomial.
Hence, it satisfies the required (quite technical) KL property, see also \cite{AttBolSva2013}.
\begin{theorem}\label{thm:ConvGD}
    Assume that $\mathcal J$ satisfies the KL property and is bounded from below.
If the sequence $(\vec z_k)_{k \in \N}$ generated by Algorithm~\ref{alg:AGD} (without the stopping criterion) is bounded, then it converges to a critical point $\hat{\vec z}$ of $\mathcal J$.
Moreover, the sequence $(\vec z_k)_{k \in \N}$ has finite length, in the sense that
\begin{equation}\label{eq:finite_length}
    \sum_k \Vert \vec z_{k+1} - \vec z_k \Vert < \infty.
\end{equation}
\end{theorem}
\begin{proof}
According to \cite[Theorem ~3.7]{OchCheBro2014}, we need to check that
\setlist[enumerate]{leftmargin=12mm, label=\roman*)}
\begin{enumerate}
    \item[(H1)]\label{item:1} there exists $a>0$ such that $\mathcal J (\vec z_{k}) + a \Vert \vec z_{k} - \vec z_{k-1} \Vert^2  \leq \mathcal J (\vec z_{k-1})$ for all $k \in \N$;
    \item[(H2)] there exists $b>0$ such that $\Vert \boldsymbol{\nabla} \mathcal J (\vec z_k) \Vert \leq 2b (\Vert \vec z_{k} - \vec z_{k-1} \Vert + \Vert \vec z_{k+1} - \vec z_{k} \Vert)$ for all $k \in \N$;
    \item[(H3)] there exists a subsequence $(\vec z_{k_j})_{j \in \N}$ such that $\vec z_{k_j} \to \vec z$ and $\mathcal J(\vec z_{k_j}) \to \mathcal J (\vec z).$ 
\end{enumerate}
These three conditions are needed in order to conclude that the iterates $(\vec z_k)_{k \in \N}$ satisfy \eqref{eq:finite_length} and converge to a critical point $\hat{\vec z}$ of $\mathcal J$.
We have already verified (H1) in \eqref{eq:EstVio} and \eqref{eq:EstNonVio}.
For (H2), we first note that, if \eqref{eq:CondCheck} is violated, then it directly holds that
\begin{align}
    \Vert \boldsymbol{\nabla} \mathcal{J}(\vec z_k) \Vert =  L\Vert \vec z_{k+1} - \vec z_k\Vert.
\end{align}
Otherwise, $\Vert \boldsymbol{\nabla} \mathcal{J}(\vec z_k) \Vert = L\Vert \vec x_{k+1} - \vec z_k\Vert$, and it follows that
\begin{align}
    \Vert \boldsymbol{\nabla} \mathcal{J}(\vec z_k) \Vert &\leq L\Vert \vec z_{k+1} - \vec z_k\Vert + \beta_k L \Vert \vec x_{k+1} - \vec x_k \Vert\notag\\
    &\leq L\Vert \vec z_{k+1} - \vec z_k\Vert + L\Bigl\Vert \vec z_k - \frac{1}{L} \boldsymbol{\nabla} \mathcal J(\vec z_k) - \vec z_{k-1} + \frac{1}{L} \boldsymbol{\nabla} \mathcal J(\vec z_{k-1})\Bigr\Vert\notag\\
    &\leq L\Vert \vec z_{k+1} - \vec z_k\Vert + 2L\Vert \vec z_k - \vec z_{k-1}\Vert.
\end{align}
Since we assume that the sequence $(\vec z_k)_{k \in \N}$ is bounded and since $\mathcal J$ is continuous, also (H3) holds and the result follows from \cite[Theorem ~3.7]{OchCheBro2014}.
\end{proof}
\begin{remark}
    To ensure that $(\vec z_k)_{k \in \N}$ remains bounded, one can simply add a regularization term $\kappa \Vert \vec x \Vert^2$  to $\mathcal{J}$, where $\kappa >0$ can be arbitrarily small.
    %, which corresponds to an additional filter and potential.
    Then, $\mathcal{J}$ becomes coercive because the profile $\psi$ has linear extensions (see Lemma \ref{lm:coercive}). Therefore, $(\vec z_k)_{k \in \N}$ must remain bounded, otherwise $(\mathcal{J}(\vec z_k))_{k \in \N}$ could not be decreasing (see \eqref{eq:EstVio}). Empirically, however, this ``trick'' was found to be unnecessary as the iterates would remain bounded in all settings explored.
\end{remark}

\begin{lemma}
\label{lm:coercive}
    Let $R$ be a ridge regularizer of the form \eqref{eq:convexridgereg}, where the profiles $\psi_j$ are continuous, even, and have affine extensions\footnote{In the sense that there exists $t_0 \in\R$ such that $\psi_j$ is affine on $(-\infty,-t_0]$ and on $[t_0, +\infty)$.}. Then
    \begin{equation}
    \label{eq:coercivelm}
        \vec x \mapsto \|\mat H \vec x - \vec y\|_2^2 + R(\vec x) + \kappa \|\vec x\|_2^2
    \end{equation}
    is coercive for any $\kappa > 0$.
\end{lemma}
\begin{proof}
    By assumption all $\psi_j$ are affine on $[t_0, +\infty)$ with slope $u_j\in\R$.
   Hence, it holds for $|t|>t_0$ that $\psi_j(t) = \psi_j(|t|) = u_j (|t| - t_0) + \psi_j(t_0)$.
   Next, we define $v_j = \min_{|t|\leq t_0} (\psi_j(t) - u_j(|t| - t_0)) \leq \psi_j(t_0)$, which is well-defined since $\psi_j$ are continuous.
   By definition of $v_j$, it holds for any $t\in\R$ that $\psi_j(t)\geq u_j(|t| - t_0) + v_j$, and the objective in \eqref{eq:coercivelm} is lower bounded by
    \begin{equation}
         \vec x \mapsto \|\mat H \vec x - \vec y\|_2^2 + \kappa \|\vec x\|_2^2 + \sum_{j=1}^{d\times N_C} u_j (|\vec w_j^T \vec x| - t_0) + v_j ,
    \end{equation}
    which is coercive for any $\kappa > 0$.
        
\end{proof}
\begin{remark}
    For $\sqrt{\lambda_{\min}(\mat H^T \mat H)} \geq \lambda$, the objective \eqref{eq:InvProbLearn} is $(\sqrt{\lambda_{\min}(\mat H^T \mat H)} - \lambda)$-strongly convex and, hence, convex.
    Then, Algorithm \ref{alg:AGD} is guaranteed to converge to a global minimum of the objective.
    Otherwise, some results on convergence to local minima come into play, including with rates\cite{AttBolSva2013,Ochs2018}.
\end{remark}
As the problem is potentially non-convex, the initialization of the algorithm may influence the final reconstruction.
However, we did not observe such a dependence in our experimental settings.
Therefore, we opted for \emph{the common zero-initialization}.
A more sophisticated strategy may take as initial configuration the reconstruction of a trustworthy convex variational model such as \cite{GouNeuBoh2022}.
Then, Algorithm~\ref{alg:AGD} would be used to refine the reconstruction.
\subsection{Experiments}
\label{subsec:exp}
The WCRR-NN model trained in Section \ref{sec:Denoising} to perform denoising on the BSD500 dataset is now deployed to solve two image-reconstruction problems using safeguarded AGD.
For each setup, $\lambda$ and $\sigma$ are tuned over a validation set to maximize the peak signal-to-noise ratio (PSNR) with the coarse-to-fine routine from \cite{GouNeuBoh2022}, and then used for evaluation.
\paragraph{MRI}
The ground truth consists of fully sampled knee images with size $(320 \times 320)$ from the fastMRI dataset \cite{zbontar2018fastMRI}.
The corresponding MRI measurements are a subsampled version of the 2D Fourier transforms ($k$-space).
This subsampling is performed with a Cartesian mask that has two parameters: the acceleration $M_{\text{acc}}=4$ and the center fraction $M_{\text{cf}}=0.08$.
All the $\lfloor 320 M_{\text{cf}} \rfloor$ columns in the center of the k-space (low frequencies) are retained in full, while columns in the other region of the k-space are uniformly sampled.
More precisely, we are left with $\lfloor 320/M_{\text{acc}} \rfloor$ selected columns.
Lastly, both the real and imaginary parts of the measurements are corrupted by Gaussian noise with standard deviation $\sigma_{\mathbf{n}} = \num{e-4}$.
For validation and testing, we picked $10$ and $99$ images, respectively, all normalized to have a maximum value of one.
\begin{table}
\centering
\caption{PSNR and SSIM values for MRI and CT reconstruction experiments.}%
\label{table:reconstruction_performance_mri}
\subfloat[MRI]{
\begin{tabular}{lcc}
\toprule
Metric & PSNR & SSIM \\
\midrule
Zero-fill & 27.92 & 0.711 \\
\midrule
TV\cite{beck2009fast} & 32.03 & 0.7922 \\
%PnP-DnCNN \cite{ryu2019plug} & 31.42 & 0.805 \\
CRR-NN \cite{GouNeuBoh2022} & {33.14} & {0.842}\\
\midrule
WCRR-NN & 34.55 & 0.858\\
Prox-DRUNet \cite{HurLec2022} & 35.09 & 0.864\\
\bottomrule\\[-1.1ex]
\end{tabular}}
\hspace{.3cm}
\subfloat[CT]{
\begin{tabular}{lccc}
\toprule
Metric & PSNR & SSIM & Param. \\
\midrule
TV & 31.57 & 0.852 & 1\\
ACR \cite{MukDit2021} & 32.17 & 0.868 & \num{6e5}\\
CRR-NN & 32.87 & 0.862 & {\textbf{\num{5e3}}} \\
\midrule
AR \cite{lunz2018adversarial} & 33.62 & 0.875 & \num{2e7}\\
WCRR-NN & 34.06 & 0.895 & {\textbf{\num{1e4}}}\\
%PnP-DnCNN \cite{ryu2019plug} & 33.83 & 0.881 & \num{6e5}\\
Prox-DRUNet & 34.20 & 0.901 & \num{2e7}\\
\bottomrule
\end{tabular}}
\end{table}

\paragraph{CT}
To provide a comparison with adversarial regularization (AR) \cite{{lunz2018adversarial}} and its convex counterpart ACR \cite{MukDit2021} (see more details in Section \ref{subsec:results_discussion}), we include the sparse-view CT experiment proposed in \cite{MukDit2021}.
Its data consist of human abdominal CT scans for 10 patients, publicly available as part of the low-dose CT Grand Challenge \cite{M2016}.
For validation, 6 images are taken uniformly from the first patient of the training set used by \cite{MukDit2021}.
To benchmark all methods, we use the same set as \cite{MukDit2021}, made of 128 slices with size $(512\times 512)$ from a single patient, all normalized to have a maximum value of one.
The CT measurements are simulated using a parallel-beam acquisition geometry with 200 angles and 400 detectors.
These measurements are corrupted by Gaussian noise with standard deviation $\sigma_{\mathbf{n}}=2.0$.
\begin{figure*}[tbp]
    \centering
    \includegraphics[width=.95\textwidth]{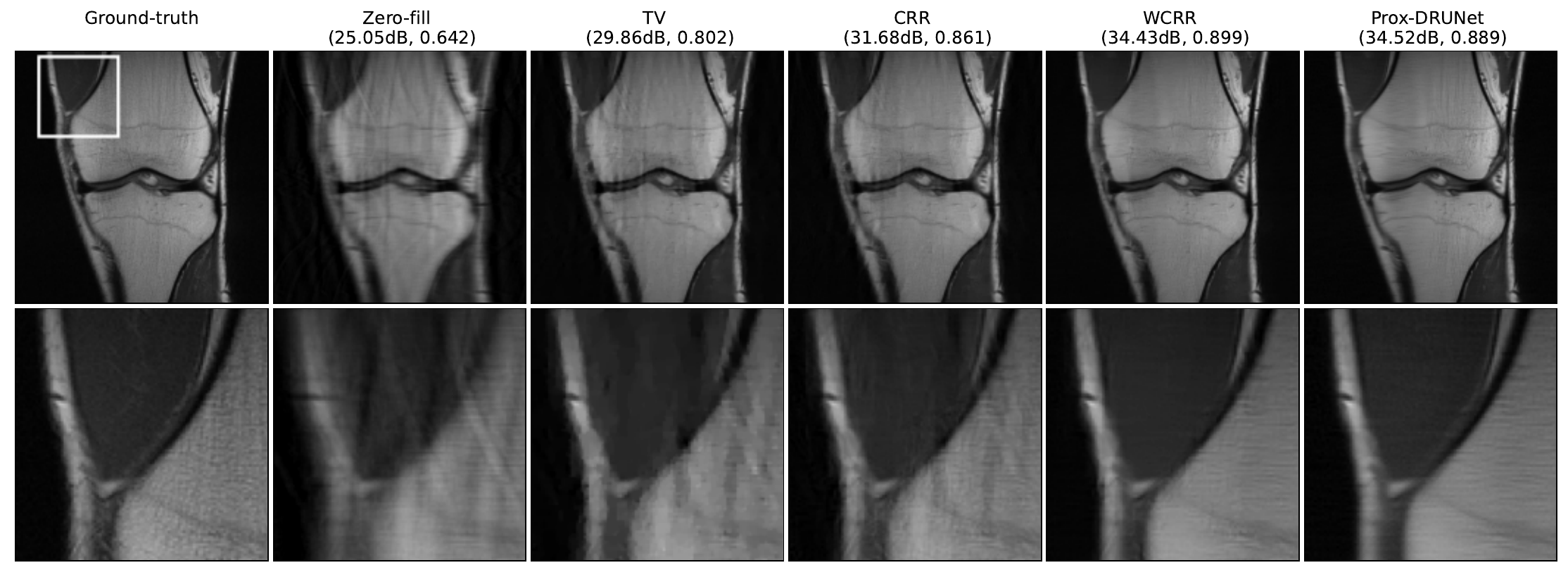}
    \caption{Reconstructions for the MRI experiment (metrics are PSNR and SSIM).}
    \label{fig:MRIreconstructions}
    \vspace{.6cm}
    \centering
    \includegraphics[width=.95\textwidth]{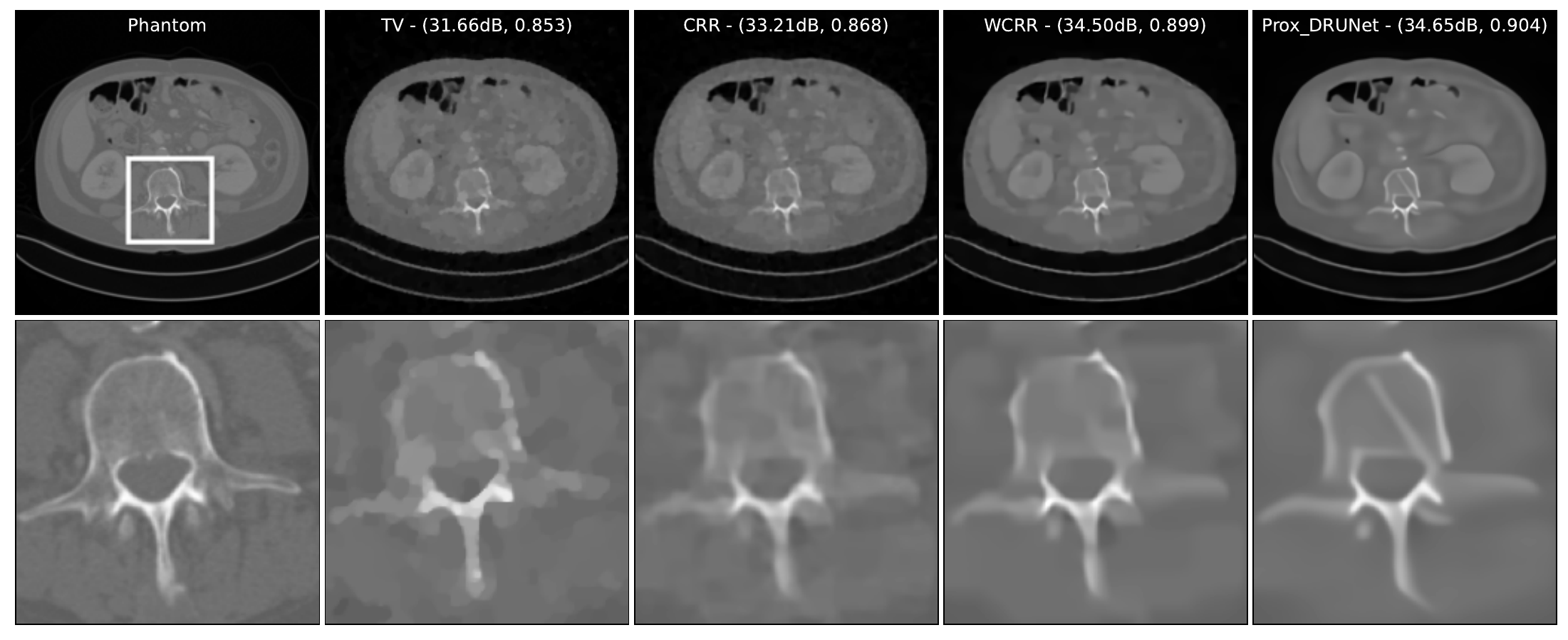}
    \caption{Reconstructions for the CT experiment (metrics are PSNR and SSIM).}
    \label{fig:CTreconstructions}
    \vspace{.6cm}
    \centering
    \includegraphics[width=0.9\textwidth]{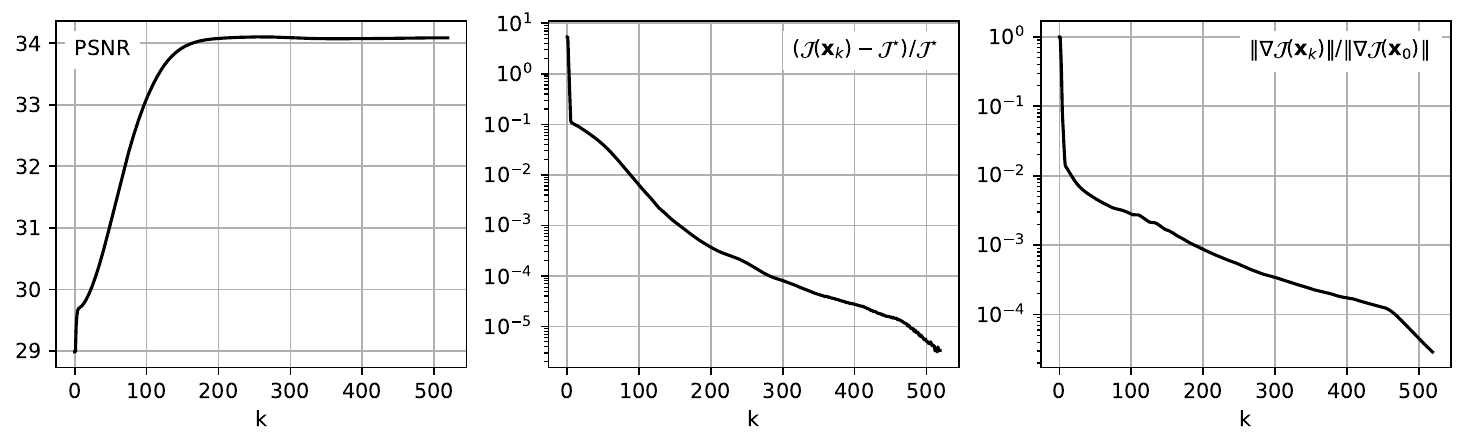}
        \caption{Example of convergence curves (MRI).}
    \label{fig:convergencemri}
\end{figure*}
\subsubsection{Comparison and Discussion}
\label{subsec:results_discussion}
The PSNR and structural similarity index measure (SSIM) values on the test sets are reported together with the parameter numbers in Table~\ref{table:reconstruction_performance_mri}.
The hyperparameters of each method are tuned to maximize the average PSNR over the validation sets with the coarse-to-fine method described in \cite{GouNeuBoh2022}.
We observe that WCRR-NNs outperform the other energy-based methods and are close to the PnP approach.
For both problems, reconstructions are provided in Figures~\ref{fig:MRIreconstructions} and~\ref{fig:CTreconstructions}, and examples of convergence curves for SAGD are given in Figures \ref{fig:convergencemri} and \ref{fig:convergencect}.
\begin{figure}[tbp]
\centering
    \includegraphics[width=0.9\textwidth]{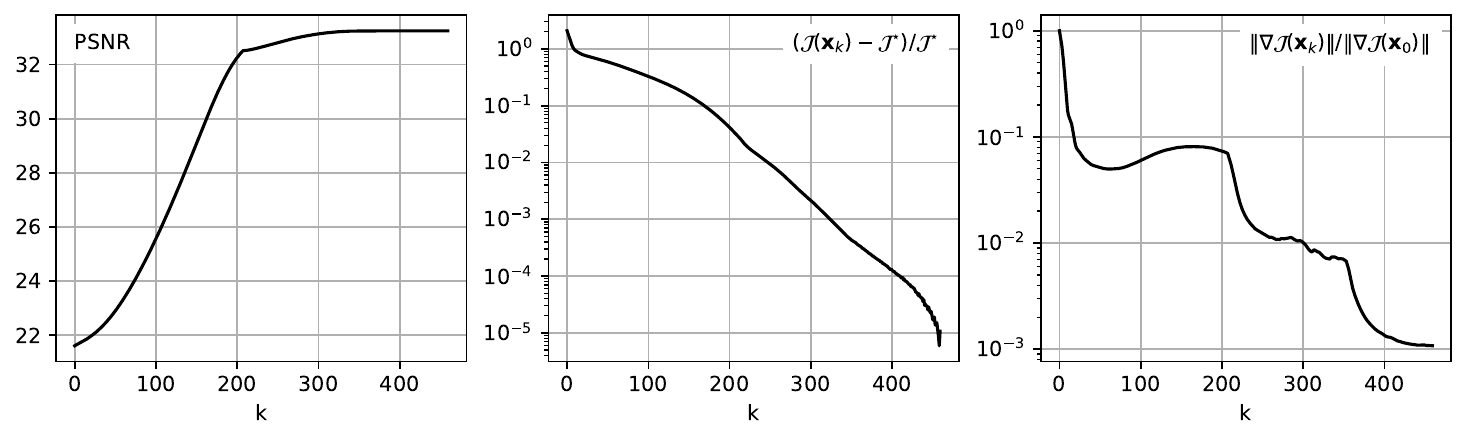}
        \caption{Example of convergence curves (CT).}
    \label{fig:convergencect}
\end{figure}

%Here, we restricted our attention to methods with open-source code and practical implementations.
Overall, the results illustrate the universality and efficiency of our method.
In the following, we briefly comment on the competing methods used in our evaluation.
\paragraph{Convex Models}
The TV and CRR-NN reconstructions serve as references for convex methods.
They are computed via the FISTA algorithm \cite{beck2009fast} with a nonnegativity constraint.
Similarly to denoising, we observe that the move from convex to weakly-convex regularization leads to significant improvements in quality.
In the MRI experiment, the aliasing artifacts introduced by CRR-NN, and even more by TV, are suppressed by the weakly convex regularizer.
In the CT experiment, the TV reconstruction includes staircasing artifacts, and CRR-NN has a slight tendency to blur the edges.
On the contrary, our weakly convex regularizer is able to produce sharp edges without blur, but sometimes at the cost of over-smoothing some background details.
%This is not surprising due to the interpretation as sparsity prior and the inherited universality.
Note that convex models are still better understood from a theoretical perspective because convergence to global optima can be guaranteed.
Hence, they might be favorable in certain settings.
\paragraph{Adversarial Regularization}
As references for explicit regularization approaches, we provide a comparison with the convex ACR~\cite{MukDit2021,MukSch2021} framework and its non-convex counterpart AR~\cite{lunz2018adversarial}.
Bypassing a gradient-based parameterization, these models parameterize the regularizer $R$ directly and train it in an adversarial manner.
As the regularizers of \cite{MukDit2021,MukSch2021,lunz2018adversarial} are tailored to a specific inverse problem, we can only provide a comparison for their CT experiment.
Even though ACR and AR have significantly more parameters than (W)CRR-NN, they perform less well.
The numerical results present favorable evidence regarding the effectiveness of the parameterization used for WCRR-NNs. Note, however, that drawing a definitive conclusion on the parameterization only is delicate since AR and ACR rely on a different training procedure.
\paragraph{Plug-and-Play}
Our approach bears some resemblance with PnP methods since $R$ is learned on a generic denoising task.
Hence, it is natural to compare WCRR-NNs with a deep CNN version of this approach.
Among countless variations, the recently proposed framework \cite{HurLec2022}, which we refer to as Prox-DRUNet, is the closest to ours in terms of theoretical guarantees and existence of an underlying regularizer (see Section~\ref{subsec:denoisingcomparison}
for a discussion).
%It is also similar to SOTA methods, although the latter have less theoretical guarantees.
We use the pre-trained DRUNet-based proximal denoiser from \cite{HurLec2022} within the PnP-PGD (proximal gradient descent) for CT and the PnP-DRS (Douglas-Rachford splitting) for MRI\footnote{PnP-DRS is well-suited to settings where the proximal operator of the data term can be efficiently computed, which includes MRI but not CT.}.
This approach, which may be considered the state of the art in energy-related PnP, yields slightly better PSNR and SSIM than our method. Note that Prox-DRUNet involves $3$ orders of magnitude more parameters and days of training.

In the MRI experiment, both Prox-DRUNet and WCRR-NN are able to avoid the aliasing artifacts typically generated by the methods that rely on a convex regularizer.
In the CT experiment, the visual inspection of the reconstructions reveals that quality metrics are only part of the story.
While the output of Prox-DRUNet always looked remarkably realistic, it was more prone to hallucination/artifact exaggeration, especially for hard problems such as the CT experiment.
In that respect, the Prox-DRUNet reconstruction in Figure~\ref{fig:CTreconstructions} is particularly telling:
It includes an elongated structure that is not present in the ground truth, nor in any other reconstruction.
While such \emph{enhanced} images are desirable in many settings and lead to state-of-the-art denoising performance, they raise major concerns for sensitive applications, including medical imaging.
Regarding the theoretical convergence guarantees of PnP-PGD, the necessary Lipschitz constraint is only enforced by regularization during training.
Unfortunately, it is infeasible to verify if it is met after training.
In practice indeed, it seems to be not fully met \cite{HurLec2022}.

\section{Conclusion}\label{sec:Conclusions}
In this paper, we proposed a method for the learning of a $1$-weakly convex regularizer that leads to a convex denoising functional.
To the best of our knowledge, this is the first instance of convex non-convex schemes that surpasses BM3D for the denoising of natural images.
A key feature of our method is that the architecture deployed to parameterize the regularizer is shallow.
Thereby, the role of each parameter is transparent:
Parameters are adjusted to produce a sparsity-promoting prior.
Although the regularization of inverse problems with the learned regularizer does not necessarily lead to a convex objective, gradient-based optimization methods are empirically effective and produce high-quality reconstructions.
In the future, a better understanding of WCRR-NNs might help to boost the performance of lightweight and robust data-driven image-reconstruction models even further.
This includes the dependence of the learned regularizer on the modality and/or on the image domain used during training.
It is indeed expected, for instance, that a fine-tuning of the regularizer with modality-specific prior knowledge will improve the quality of the reconstruction.

\section*{Acknowledgments}
The research leading to this publication was supported by the European Research Council (ERC) under European Union’s Horizon 2020 (H2020), Grant Agreement - Project No 101020573 FunLearn, and by the Swiss National Science Foundation, Grant 200020 184646/1.
The authors are thankful
to Pakshal Bohra and Stanislas Ducotterd for helpful discussions.
Finally, the authors want to thank the anonymous reviewers for their valuable comments.
\bibliographystyle{abbrv}
\bibliography{references,references_cvx}
\end{document}